\documentclass[12pt]{article}
\usepackage{graphicx}
\usepackage{verbatim}
\usepackage{amsmath,amsthm,amsfonts}
\usepackage{amssymb}

\newtheorem{theorem}{Theorem}[section]
\newtheorem{proposition}[theorem]{Proposition}

\newtheorem{lemma}[theorem]{Lemma}
\newtheorem{corollary}[theorem]{Corollary}

\theoremstyle{definition}

\newtheorem{def-theorem}[theorem]{Definition-Theorem}
\newtheorem{remark}[theorem]{Remark}

\newtheorem{example}[theorem]{Example}
\numberwithin{equation}{section}

\newcommand{\be}{\begin{equation}}
\newcommand{\LME}{{\cal S}_{LME}}
\newcommand{\ee}{\end{equation}}
\newcommand{\bea}{\begin{eqnarray}}
\newcommand{\eea}{\end{eqnarray}}
\newcommand{\beas}{\begin{eqnarray*}}
\newcommand{\eeas}{\end{eqnarray*}}
\newcommand{\ba}{\begin{array}}
\newcommand{\ea}{\end{array}}
\def\identity{{\rlap{1} \hskip 1.6pt \hbox{1}}}
\newcommand{\tr}{{\rm tr}}
\newcommand{\GITquot}{/\!/}
\newcommand{\SL}{\operatorname{SL}}
\newcommand{\SU}{\operatorname{SU}}
\newcommand{\GL}{\operatorname{GL}}
\newcommand{\bbP}{\mathbb{P}}

\begin{document}

\begin{titlepage}
\hfill
\vbox{
    \halign{#\hfil \cr
    MIT-CTP/4977  \cr
        } 
      }  
\vspace*{20mm}
\begin{center}
{\Large \bf Locally Maximally Entangled States \\ of Multipart Quantum Systems}

\vspace*{15mm}
\vspace*{1mm}
 Jim Bryan$^a$, Samuel Leutheusser$^{b,c}$, Zinovy Reichstein$^a$, and Mark Van Raamsdonk$^b$
\vspace*{1cm}
\let\thefootnote\relax\footnote{sam.leutheusser@alumni.ubc.ca, jbryan@math.ubc.ca, reichst@math.ubc.ca, mav@phas.ubc.ca}

{${}^{a}$ Department of Mathematics,
University of British Columbia\\
1984 Mathematics Road,
Vancouver, B.C., V6T 1Z1, Canada\\
${}^{b}$ Department of Physics and Astronomy,
University of British Columbia\\
6224 Agricultural Road,
Vancouver, B.C., V6T 1Z2, Canada \\
${}^{c}$
Center for Theoretical Physics,
Massachusetts Institute of Technology,
77 Massachusetts Avenue,
Cambridge, MA 02139, USA
}

\end{center}
\begin{abstract}

For a multipart quantum system, a locally maximally entangled (LME)
state is one where each elementary subsystem is maximally entangled
with its complement. This paper is a sequel to~\cite{mathpaper},
which gives necessary and sufficient conditions for a system to admit LME states in terms of its subsystem dimensions $(d_1, d_2, \dots, d_n)$, and computes the dimension of the space $\LME/K$ of LME states up to local unitary transformations for all non-empty cases. Here we provide a pedagogical overview and physical interpretation of the underlying mathematics that leads to these results and give a large class of explicit constructions for LME states. In particular, we construct all LME states for tripartite systems with subsystem dimensions $(2,A,B)$ and give a general representation-theoretic construction for a special class of stabilizer LME states. The latter construction provides a common framework for many known LME states. Our results have direct implications for the problem of characterizing SLOCC equivalence classes of quantum states, since points in $\LME/K$ correspond to natural families of SLOCC classes. Finally, we give the dimension of the stabilizer subgroup $S \subset \SL(d_1, \mathbb{C}) \times \cdots \times \SL(d_n, \mathbb{C})$ for a generic state in an arbitrary multipart system and identify all cases where this stabilizer is trivial.

\end{abstract}

\end{titlepage}
\tableofcontents

\setcounter{footnote}{0}

\section{Introduction}

Consider a multipart quantum system whose pure states are vectors in a tensor product Hilbert space
\[
{\cal H} = {\cal H}_{1}  \otimes \cdots  \otimes {\cal H}_{n}
\]
with subsystems ${\cal H}_i$ of dimension $d_i$. We will say that a state in ${\cal H}$ is locally maximally entangled (LME) if the reduced density matrix corresponding to each elementary subsystem is a scalar multiple of the identity operator on that subsystem.\footnote{Such states are sometimes called
locally maximally mixed states or 1-uniform states. Various authors in the past (see,
for example \cite{GB, HS,Scott, GZ, Helwig:2012nha, FFPP, HGS}) have also considered
the more constrained problem of requiring in addition that certain composite subsystems are maximally mixed \cite{Scott}. LME states are called $m$-uniform when all $m$-party reduced density matrices are maximally mixed. States for which all subsystems with dimension less than or equal to the dimension of the complementary subsystem are called absolutely maximally entangled (AME) or maximally multipartite entangled states (MMES) \cite{FFPP}. We will discuss and provide examples of these more special states in section 3.1, examples 7 and 8.} Fixing the dimension vector $(d_1,\dots,d_n)$, we define $\LME \subset {\cal H}$ to be the subset of states which are locally maximally entangled:
\be
\LME = \left\{|\Psi \rangle \in {\cal H} \; \Big| \; \rho_i \equiv \tr_{\bar{i}}|\Psi \rangle \langle \Psi | = \frac{1}{d_i} \identity \right\}\; .
\ee
LME states play an important role in many applications of quantum mechanics and quantum information
theory. Their importance has been pointed out by many authors in the past, for example \cite{Kly02, VDD, Scott}.
Well-known examples of LME states include Bell states
\be
\label{Bell}
|\Psi \rangle = \frac{1}{\sqrt{d}} \sum_{i} |i \rangle  \otimes |i \rangle \in {\cal H}_d  \otimes {\cal H}_d \; ,
\ee
the GHZ state
\be
\label{GHZ}
|\Psi \rangle = \frac{1}{\sqrt{2}} (|0 \rangle  \otimes |0 \rangle  \otimes |0 \rangle  + |1 \rangle  \otimes |1 \rangle  \otimes |1 \rangle) \in {\cal H}_2  \otimes {\cal H}_2  \otimes {\cal H}_2
\ee
and its generalizations, various quantum error-correcting code states, cluster states, and graph states.

In this paper, we consider the following basic questions:
\begin{enumerate}
\item
For which dimensions $(d_1,\dots,d_n)$ do LME states exist?\footnote{It has been suggested previously that the inequalities $d_i \le \prod_{j \ne i} d_j$ provide necessary and sufficient conditions. We will see that these conditions are necessary but not sufficient.}
\item
How can we characterize the space of these states? What is the dimension, geometry, etc... ?
\item
Can we give explicit constructions of states in $\LME$?
\end{enumerate}
These questions exhibit a remarkable mathematical richness: it turns out that they are related to natural questions in representation theory, symplectic and algebraic geometry, and geometric invariant theory, as reviewed for example in \cite{Walter}, \cite{SOK14} (see additional reference below). Making use of tools from all of these areas, we are able to provide a complete answer to (1), and new results for questions (2) and (3).

\subsubsection*{Explicit constructions}

We begin in section 2 with some explicit constructions, reviewing the
case $n=2$ and considering in detail the case $n=3$ with dimensions $2
\le d_2 \le d_3$. In these cases, we can explicitly construct all
states in $\LME$. For $n=2$ these states exist if and only
if $(d_1,d_2) = (N,N)$ while for $(d_1, d_2, d_3) = (2,d_2,d_3)$ we
find that these states exist if and only if $(d_2,d_3) = (N,N)$ with
$N \ge 2$ or $(d_2,d_3) = (Nk,(N+1)k)$ with $N k \ge 2$. Except for the
cases $(2,N,N)$ with $N \ge 4$, we find that there is a unique LME
state up to local unitary transformations (i.e. up to a change of
basis for each subsystem). For the case $(2,N,N)$ with $N \ge 3$ the
space of LME states up to local unitary transformations has real
dimension $2(N-3)$. It is equivalent to the space of collections of
$N$ unit vectors in $\mathbb{R}^3$ adding to zero, with sets of
vectors related by an $SO(3)$ rotation considered as
equivalent.\footnote{Remarkably, this space is also equivalent to $N$-tuples of
points on $\mathbb{CP}^{1}$, with no more than $N/2$ of the points
coinciding, up to M\"obius transformations. This is a very concrete
realization of the Kempf-Ness equivalence between symplectic
and algebraic quotients which we discuss in subsequent sections.}

\subsubsection*{Connection to representation theory}

In section 3, we point out a general way to construct LME states using
data from the representation theory of finite or connected compact
groups. Given any group $H$ with unitary irreducible representations
$R_1,\dots ,R_n$ of dimensions $d_1, \dots, d_n$ whose tensor product
contains the trivial representation, let \be |0 \rangle =
\sum_{\vec{i}} C_{i_1 \cdots i_n} |i_1 \rangle  \otimes \cdots  \otimes
|i_n \rangle \ee be an explicit representation of a vector in $V_{d_1}
 \otimes \cdots  \otimes V_{d_n}$ which transforms trivially. Then we
can show that the vector $|0 \rangle$, viewed as a quantum state in
the Hilbert space ${\cal H}$, is locally maximally entangled. Many of
the well-known LME states can be understood in this way. For example,
for any group $H$, the trivial representation obtained from the tensor
product of any representation and its dual gives a Bell state, while
the GHZ state corresponds to the trivial representation obtained in
the tensor product of three two-dimensional representations of the
symmetric group $S_3$.

We show that all LME states obtained through this construction can be represented via tensor networks corresponding to tree graphs with trivalent vertices, where the edges are labeled with representations of $H$ and the tensors corresponding to each vertex are Clebsch-Gordon like coefficients which construct one representation from the product of the other two. We point out (based on an observation of Eliot Hijano) that it is possible to construct quantum systems for which all states are of this type by starting with a multipart quantum system whose symmetry group $H$ acts irreducibly on each subsystem and then gauging $H$.

Finally, we point out that our results provide some insight into an open question in representation theory of finding the set of all dimensions $\{d_1, d_2, d_3\}$ for which there exist (for some group) unitary irreducible representations of dimension $d_1$ and $d_2$ whose tensor product contains a representation of dimension $d_3$.

\subsubsection*{General characterization of $\LME$}

In sections 4 and 5, we move on to the question of
characterizing the space $\LME$ in general.
Remarkably, the question of characterizing LME states may be phrased as a natural question
in symplectic geometry, and then as an equivalent question in algebraic geometry / geometric
invariant theory.  These connections are well documented in the literature; see
for example~\cite{Kly02}, \cite{VDD}, \cite[\S3]{Kly07}, \cite[\S4]{wallach}, \cite{SOK12}
(and \cite{Walter} or \cite{SOK14} for recent reviews).

 \begin{figure}
\centering
\includegraphics[width=0.8\textwidth]{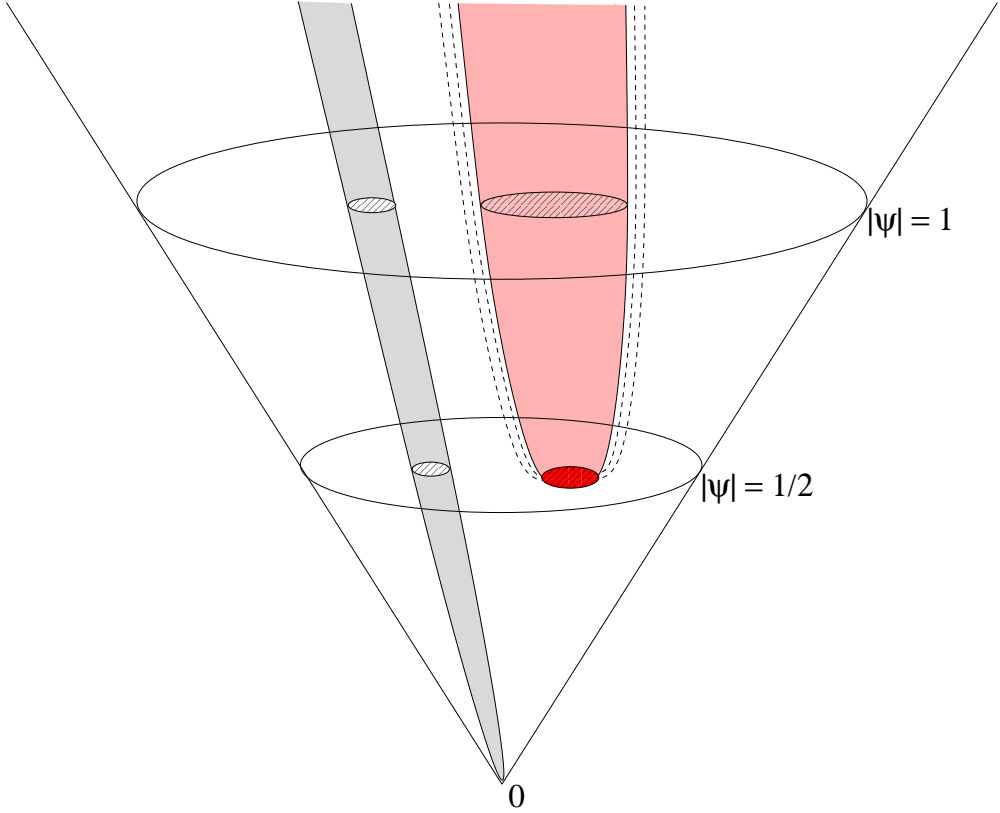}
\caption{Correspondence between LME states and SLOCC equivalence classes. The cone denotes the full Hilbert space ${\cal H}$ of unnormalized states. Horizontal sections correspond to states with specific normalization. These sections decompose into orbits of $K = \SU(d_1) \times \cdots \times \SU(d_n)$ (crosshatched horizontal ellipses). SLOCC equivalence classes correspond to orbits of $G = \SL(d_1, \mathbb{C}) \times \cdots \times \SL(d_n,\mathbb{C})$ (vertical); these are made up of $K$-orbits with a range of normalizations. LME states are precisely those which minimize the norm on their orbit.  Orbits with such states are closed and called {\it polystable} (e.g. solid orbit on the right). Each polystable orbit contains a single $K$-orbit of LME states (red ellipse). Any orbit with a lower bound on the norm (known as a {\it semistable} orbit) contains a single polystable orbit in its closure. Orbits with the same polystable orbit in their closure are considered equivalent. Thus, $K$-orbits of LME states are in one-to-one correspondence with equivalence classes of semistable $G$-orbits (e.g. the family denoted by dashed curves). Each equivalence class corresponds to a  natural family of SLOCC orbits, and these families encompass all states apart from those (called {\it unstable}) in orbits containing states of arbitrarily small norm (e.g. the orbit shown on the left). Mathematically, the space of these equivalence classes is known as the {\it geometric invariant theory} (GIT) quotient ${\cal H}\GITquot G$. See section 4 for a detailed review. }
\label{fig:cone}
\end{figure}

Through these connections, the space $\LME/K$ of LME states with equivalence under local unitary transformations $K = U(d_1) \times \cdots \times U(d_n)$ is identified with the ``geometric invariant theory'' (GIT) quotient of the full space of states $\bbP({\cal H})$
by the larger group $G = \SL(d_1, \mathbb{C}) \times \cdots \times \SL(d_n, \mathbb{C})$.
\footnote{Here $\bbP({\cal H})$ denotes the space of normalized states with equivalence up to phase;
mathematically, this is the projective space $\mathbb{CP}^{d_1 \cdots d_n -1}$.}
Note that $G$ is the complexification of $K$.
The GIT quotient has the structure of a K\"ahler manifold and can be described as a complex algebraic variety in a weighted projective space.

Physically, the group $G$ corresponds to stochastic local operations and classical communications (SLOCC) \cite{DVC}. Orbits under $G$ are SLOCC equivalence classes of quantum states.  The GIT quotient represents the space of these equivalence classes for states with ``generic'' entanglement,\footnote{Here, generic states are those which cannot be transformed via SLOCC transformations to vectors with arbitrarily small norm; equivalently, they can be distinguished from the 0 vector by some SLOCC invariant polynomial.} with the additional identification of orbits if one orbit is in the closure of another - these identifications group natural families of SLOCC classes \cite{SOK14}. A pedagogical introduction to the beautiful mathematics behind this equivalence is provided in section 4; some key points are summarized in figure \ref{fig:cone}. The presence of LME states as special points on the $G$-orbits of more general states was understood by \cite{VDD} and led these authors to refer to the LME states as ``normal forms,'' which is still a widely used terminology.

In section 5, we use the equivalence and the tools of geometric invariant theory to provide necessary and sufficient conditions on the dimensions $\vec{d} = (d_1, \dots, d_n)$ that determine whether the space $\LME$ is non-empty, and we determine the dimension of $\LME/K$ in the non-empty cases. This section provides a non-technical  overview of the rigorous mathematical results presented in \cite{mathpaper}.

To state the existence condition, we define for positive integers $k_i$
\be
\label{defN}
N(k_1, \dots, k_n) = \sum_{l=1}^n (-1)^{l+1} \sum_{1 \le i_1 < \cdots < i_l \le n} (\gcd(k_{i_1},\dots,k_{i_l}))
\ee
which is equal to the number of rational numbers in $(0,1]$ whose denominator divides some $k_i$.\footnote{This interpretation of $N$ has been pointed out to us by David Savitt.} Then we have:
\begin{theorem}\label{thm.R}
For a multipart Hilbert space with elementary subspace dimensions $\vec{d} = (d_1, \dots, d_n)$ there exist locally maximally entangled states if and only if $R(\vec{d}) \ge 0$, where
\be
R(\vec{d}) = \prod_i d_i - N(d_1^2,\dots,d_n^2) \; .
\ee
\end{theorem}
To state the result for the dimension of $\LME/K$ in the non-empty cases,  we define
\be
\label{gmd}
{\rm g}_{max}(\vec{d}) \equiv \max_{1 \le i < j \le n} \gcd(d_i,d_j)
\ee
and define the {\it expected (complex) dimension}
\begin{equation} \label{defDelta}
\Delta(\vec{d}) \equiv \left(\prod_i d_i - 1\right) - \sum_i (d_i^2 - 1) = \dim(\bbP({\cal H})) - \dim(G) \; .
\end{equation}
Then our result for the dimension of $\LME/K$ is given by:
\begin{theorem}\label{thm.dims}
Let ${\cal H}$ be a multipart Hilbert space with elementary subspace dimensions $\vec{d} = (d_1, \dots, d_n)$. Then:

If $\Delta(\vec{d}) > -2$, then $R > 0$ and $\dim(\LME/K) = \Delta(\vec{d}) \ge 0$.

If $\Delta(\vec{d}) = -2$, then $R > 0$ and $\dim(\LME/K)  = \max ({\rm g}_{max}(\vec{d}) -3,0)$.

If $\Delta(\vec{d}) < -2$, then $R \le 0$ and $\LME/K$ is a single point for $R=0$ and empty for $R<0$.
\end{theorem}

Both of these results follow from a recursive algorithm for the dimension of $\LME/K$ given as Theorem \ref{thm.States} below. Using this algorithm, we can also provide some more explicit results
for the dimension of $\LME/K$.

\begin{figure}
\centering
\includegraphics[width=0.7\textwidth]{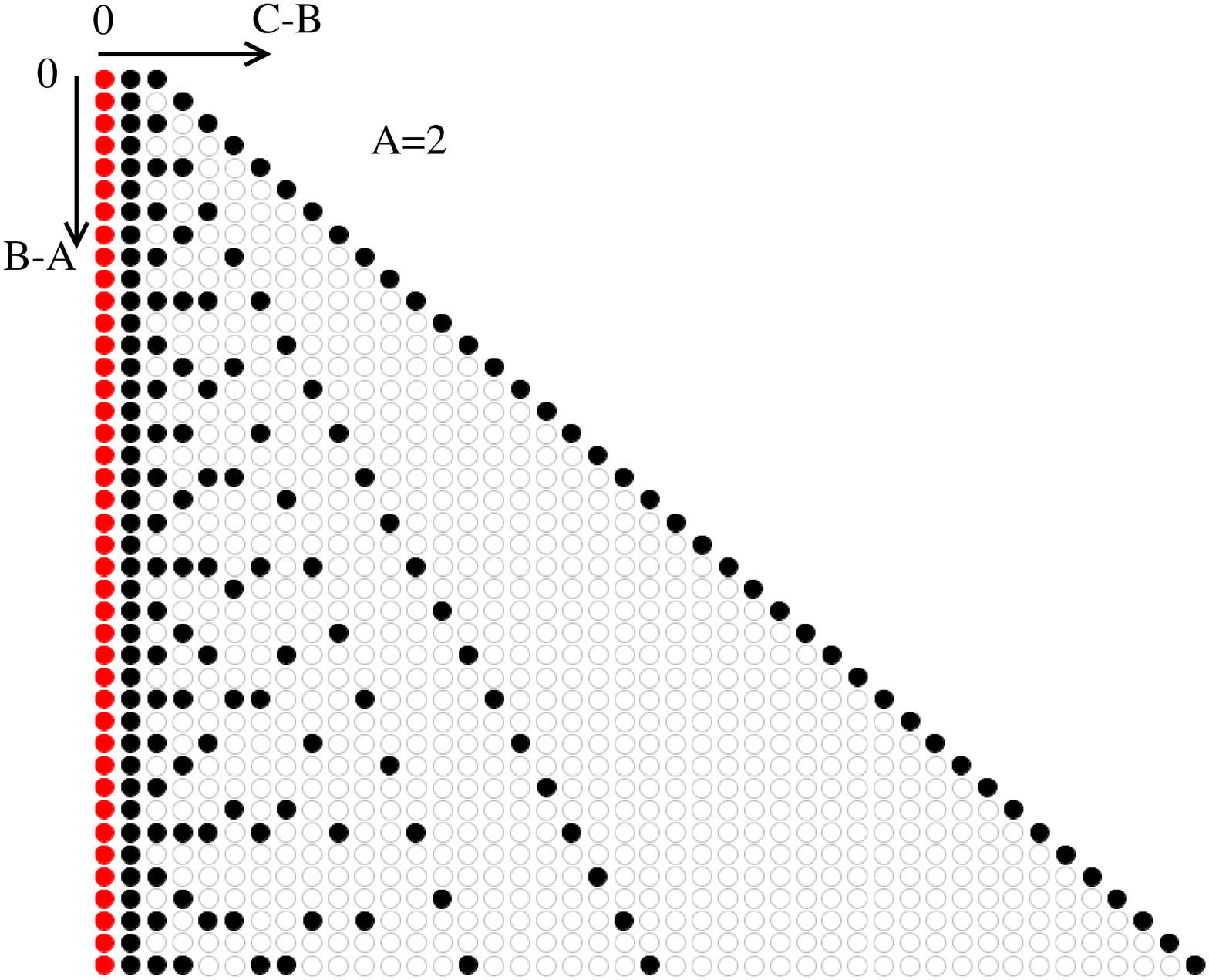}
\includegraphics[width=0.7\textwidth]{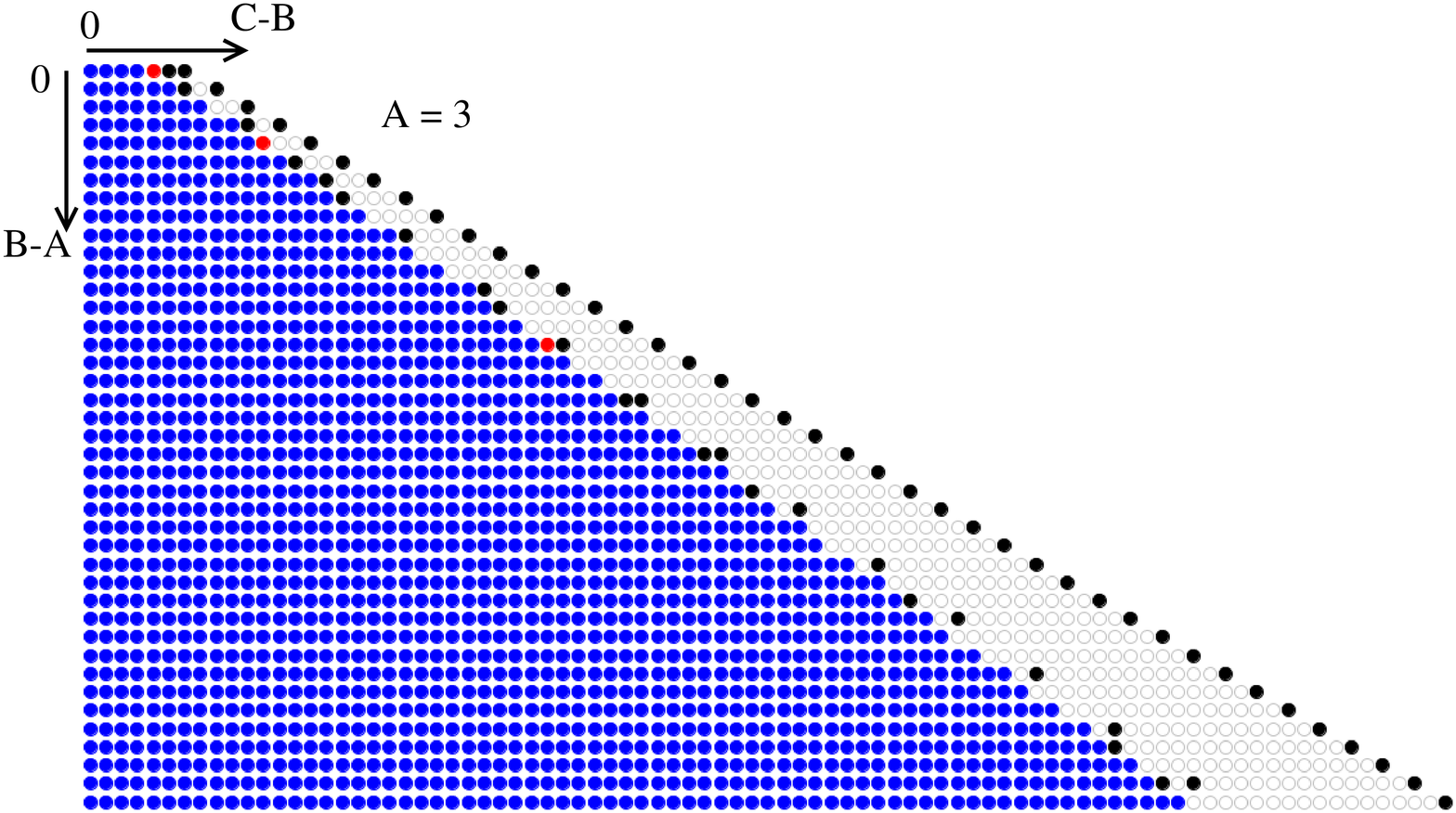}
\caption{Properties of the space $\LME/K$ of locally maximally entangled states up to local unitary transfomations for tripartite systems with dimensions $(A=2,B,C)$ and $(A=3,B,C)$. In each plot, circles correspond to values of $(A,B,C)$ (plotted with axes $(B-A)$ and $(C-B)$ for fixed $A$) for which the necessary conditions $B \le C \le AB$ for the existence of LME states are satisfied. Empty circles indicate cases where LME states do not exist. Blue filled circles (on the left in the $A=3$ case) indicate cases where the dimension of $\LME/K$ is equal to the expected dimension $ABC - A^2 - B^2 - C^2 + 2$ of the quotient $\bbP({\cal H})/G$. Black circles are non-empty cases with negative expected dimension but for which $\LME/K$ is a single point. Red circles are non-empty cases with expected dimension -2 for which $\LME/K$ is equivalent to a K\"ahler manifold describing unit vectors in $\mathbb{R}^3$ adding to zero with equivalence under $SO(3)$ rotations.}
\label{fig:23BC}
\end{figure}

Generally, for fixed
$(d_1,\dots, d_{n-1})$, there is a range of values $d_{n-1} \le d_n <
d_*$ for which the $\LME/K$ is not empty and its dimension $\Delta(\vec{d})$
is the naive dimension of the quotient $\bbP({\cal H})\GITquot G$. For $d_* \le d_n \le d_1 \cdots d_n$,
this naive dimension is negative,\footnote{Specifically, $d_*$ is
defined by setting $\Delta = -2$.} but there are sporadic values of
$d_n$ in this range for which the space $\LME/K$ is
nevertheless non-empty; it is given by a single point unless $\Delta =
-2$. For $n=3$, we are able to characterize the sporadic cases more explicitly. We find that they are all of the form $(d_1,d_2,d_3) = (A, f_i, f_{i+1})$ where $(f_i,f_{i+1})$ are successive terms in a Fibonnaci-like sequence defined by
\be
\label{Fib0}
f_{i+1} = A f_i - f_{i-1}
\ee
with $(f_0,f_1) = (k,kA)$ for some positive integer $k$ or $(f_0,f_1) \in S_A$, where $S_A$ is a finite set of pairs defined in Proposition 5.4. The results for dimensions $(2,B,C)$ and $(3,B,C)$ are displayed in figure \ref{fig:23BC}.

\subsubsection*{Relation to the quantum marginal problem and its solution}

Deciding whether LME states exist for given subspace dimensions is a particular case of the {\it quantum marginal problem}: given a collection of density operators associated to subsystems a of a multipart system, can these arise as the reduced density operators from a quantum state of the full system? This problem has received considerable attention in the past and has applications in quantum chemistry (see, for example \cite{TEVS, TVS}).

In the case where the subsystems do not overlap and the full system is assumed to be in a pure state, a general answer has been provided by Klyachko \cite{Kly04}.\footnote{See \cite{Walter} for a review.
For a discussion of the quantum marginal problem using a different approach from Klyachko's,
see~\cite{MT, VW} and references therein.} Klyachko's criterion is a set of inequalities on the spectra for the density operators, which may be expressed in the language of representation theory of the symmetric group (see section 4 for more details). These results provide an in-principle method to answer our question, but this becomes computationally intractable as the subsystem dimensions increase. For our specific instance of the quantum marginal problem, our existence results provide a much more direct and useful answer. It would be very interesting to understand whether similar methods could be used to address more general instances of the quantum marginal problem.

\subsubsection*{Multipart systems whose generic states have trivial stabilizer}

As we described above, the space of LME states up to local unitary transformations is isomorphic to the space of SLOCC orbit families if we exclude states which can be mapped arbitrarily close to the zero vector by SLOCC transformations $G$. In characterizing the dimension of this space of SLOCC families, an important step is to understand the stabilizer for a generic vector in the Hilbert space; that is, the subgroup of $G = \SL(d_1, \mathbb{C}) \times \cdots \times \SL(d_n, \mathbb{C})$ which leaves the vector invariant. Indeed, as we review in section 5, the dimension of the stabilizer subgroup $S$ for a generic vector gives the difference between the actual dimension of $\LME/K \sim (\bbP({\cal H}) \GITquot G)$ and the naive dimension $\Delta = {\rm dim}(\bbP({\cal H})) - {\rm dim}(G)$ of $\bbP({\cal H})/G$:
\be
{\rm dim}(S) = \Delta - {\rm dim}(\LME/K) \; ,
\ee
so finding the dimension of $\LME/K$ is equivalent to finding the dimension of the stabilizer of a generic vector.

There has been a significant amount of previous work on characterizing the stabilizer in the mathematics literature (see~\cite{elashvili, ampopov}), and we make use of a number of these results in \cite{mathpaper} in order to compute the dimension of the stabilizer of a generic state in all possible cases. As an aside, we point out here that this literature can also be used to answer another question of significant physical interest (see \cite{SWGK} for a recent discussion): for which multipart systems does the generic state have a trivial stabilizer (i.e. the vector is unfixed by any nontrivial SLOCC transformation)? As emphasized in \cite{GKW, SWGK}, this is interesting because it implies that almost all states in such systems are isolated in the sense that they cannot be related by deterministic LOCC transformations.

Using Theorem 2 of \cite{ampopov} and Corollary 1 to Lemma 2 of \cite{elashvili} together with results in our paper \cite{mathpaper}, we can show (see appendix C) that the stabilizer of a generic vector in Hilbert space will be trivial (i.e. the set of states with trivial stabilizer is open and dense in the full set of states) if and only if
\be
\Delta(\vec{d}) > 3
\ee
where $\Delta$ is the expected dimension $\Delta(\vec{d})$ defined in (\ref{defDelta}).\footnote{In these cases, the expected dimension is the actual dimension of $\LME/K \sim \bbP({\cal H})\GITquot G$, while in other cases the actual dimension is always less than or equal to the expected dimension, so we can also say that the stabilizer of a generic state is trivial if and only if the complex dimension of the space  $\LME/K \sim \bbP({\cal H})\GITquot G$ is larger than 3.}

For qubit systems, this holds whenever $n>4$, so we reproduce the results in \cite{GKW}. For systems of $n$ qutrits, this holds for $n > 3$. The condition also holds for any system of three or more subsystems of the same dimension $d > 3$. These cases reproduce and generalize the results in \cite{SWGK}. Our general condition covers all remaining cases, where the subsystem dimensions are not equal.

\section{Setup and simple cases}

We will consider pure states in a multipart Hilbert space
\[
{\cal H} = {\cal H}_{1}  \otimes \cdots  \otimes {\cal H}_{n} \;
\]
with subsystems ${\cal H}_i$ of dimension $d_i$. We can write a general state explicitly as
\be
|\Psi \rangle = \sum_i \psi_{i_1 \cdots i_n} |i_1 \rangle  \otimes \cdots  \otimes |i_n \rangle,
\ee
where states $|i_k \rangle$ with $1 \le i_k \le d_k$ form an orthonormal basis of ${\cal H}_k$. The density matrix for the $k$th subsystem is given by $\rho_k = \tr_{\bar{k}} | \Psi \rangle \langle \Psi |$ or explicitly as
\be
\rho_{j_k} {}^{l_k} = \sum_{i}  \psi_{i_1 \cdots j_k \cdots i_n}(\psi^*)^{i_1 \cdots l_k \cdots i_n}
\ee
A locally maximally entangled state $|\Psi \rangle$ is defined as a state such that for every $k$ we have
\be
\rho_{j_k} {}^{l_k} = \frac{1}{d_k} \delta_{j_k} {}^{l_k} \; .
\ee

\subsection{The Schmidt decomposition, bipartite systems, general necessary conditions}

For a bipartite system with ${\cal H} = {\cal H}_A  \otimes {\cal H}_{\bar{A}}$, we can write a general (pure) state using the Schmidt decomposition as
\be
|\Psi \rangle = \sum_i \sqrt{p_i} |\psi_i^A \rangle  \otimes |\psi_i^{\bar{A}} \rangle,
\ee
where $|\psi_i^A\rangle$ are orthonormal states in ${\cal H}_A$, $|\psi_i^{\bar{A}} \rangle$ are orthonormal states in ${\cal H}_{\bar{A}}$, and $p_i$ are positive real numbers with $\sum_i p_i = 1$. The density operator for the first subsystem is
\be
\rho_1 = \sum_i p_i |\psi_i^A \rangle \langle \psi_i^A| \; .
\ee
This is a multiple of the identity operator if and only if $\{|\psi_i^A \rangle\}$ form an orthonormal basis of ${\cal H}_A$ and $p_i = \dfrac{1}{ d_A}$ for all $i$. This is possible only if the dimension of ${\cal H}_A$ is less than or equal to the dimension of ${\cal H}_{\bar{A}}$. Otherwise, we can't have the required $d_A$ orthogonal states $\{|\psi_i^A \rangle\}$ since the number of these is equal to the number of orthonormal states $|\psi_i^{\bar{A}} \rangle$ which is limited by the dimension of ${\cal H}_{\bar{A}}$.

We can consider a general system with dimensions $(d_1,\dots, d_n)$ as a bipartite system with ${\cal H}_A = {\cal H}_k$ and ${\cal H}_{\bar{A}} =  \otimes_{i \ne k} {\cal H}_i$. Then a state $|\Psi \rangle \in {\cal H} =  \otimes {\cal H}_i$ can be LME only if $\rho_k \equiv \rho_A$ is proportional to the identity; we have just seen that this requires $d_A \le d_{\bar{A}}$, so we arrive at general necessary conditions
\be
\label{necessary}
d_k \le \prod_{i \ne k} d_i
\ee
for the existence of locally maximally mixed states. It has been suggested that these conditions are also sufficient, but we will see that this is not the case.

\subsection{Explicit construction: $n=2$}

For a two-part system, the necessary conditions (\ref{necessary}) give
immediately that $d_1 = d_2 \equiv d$ in order for there to exist LME
states. The discussion in the previous subsection also implies that
such a state can be written via the Schmidt decomposition as
\be
\label{LME2} |\Psi \rangle = \frac{1}{\sqrt{d}} \sum_i |\psi_i^1
\rangle  \otimes |\psi_i^2 \rangle \ee
where $|\psi_i^1 \rangle$ and
$|\psi_i^2 \rangle$ are orthonormal bases of ${\cal H}_1$ and ${\cal
H}_2$. Any state of this form is in $\LME$. The group of
local unitary transformations in $K = \SU(d) \times \SU(d)$ allow us to
independently rotate the bases for the two subsystems, so any LME
state (\ref{LME2}) can always be written as $U_1  \otimes U_2 |\Psi
\rangle_{Bell}$ where $|\Psi \rangle_{Bell}$ is the Bell state
(\ref{Bell}). We thus have the well-known result that for $n=2$, LME
states exist if and only if $d_1 = d_2$, and in this case, the Bell
state is the unique LME state up to local unitary transformations.

\subsection{Explicit construction: $n=3$, $2 = d_1 \le d_2 \le d_3$}\label{sec:Explicit23}

As a more interesting example, consider the case with $n=3$ where the
first subsystem is a qubit (i.e. $d_1 = 2$). We will assume without
loss of generality that $(d_1, d_2, d_3) = (2,B,C)$ with $2 \le B \le
C$. In this case, the necessary conditions (\ref{necessary}) require
that \be B \le C \le 2 B \ee for the existence of locally maximally
mixed states. This case is again simple enough that we can explicitly
solve the algebraic equations on the coefficients $\psi_{i_1 i_2 i_3}$
obtained by demanding that each subsystem is maximally mixed. Our
detailed analysis is presented in appendix $A$. The result of this
analysis is that LME states exist if and only if $C=B$ or $(B,C) =
(NK,(N+1)K)$. This also follows from the general result Theorem~\ref{thm.R}, as we explain in example 5.6 below.

For $B=C$, we can describe all possible LME states in terms of a set
$\{\vec{n}_i \}$ of $B$ unit vectors in $\mathbb{R}^3$ adding to zero. We
associate to $\vec{n}_i$ the qubit state $|\vec{n}_i \rangle$ for
which $S_{\vec{n}_i} |\vec{n}_i \rangle = (\hbar/2)|\vec{n}_i \rangle$, where $S_{\vec{n}_i} = \vec{n} \cdot \vec{\sigma}$ is the spin operator associated with direction $\vec{n}$. Then the state
\be \frac{1}{\sqrt{B}} \sum_{i=1}^B
|n_{i} \rangle  \otimes |i \rangle  \otimes |i \rangle \ee
is LME, and
all LME states for $(d_1,d_2,d_3) = (2,B,B)$ can be obtained from
states of this form by local unitary transformations. Transformations
which simultaneously rotate all the vectors in $\mathbb{R}^3$ also
correspond to local unitary transformations, so the space of LME
states is equivalent to the space of $B$ unit vectors in
$\mathbb{R}^3$ adding to zero with equivalence under $SO(3)$. This
gives a unique state up to local unitary transformations for the cases
$B=2$ and $B=3$, and a space of real dimension $2(B-3)$ for $B \ge 3$.

For $(B,C)=(NK,(N+1)K)$, the analysis in appendix $A$ shows that there is a unique LME state up to local unitary transformations. In the case $K=1$, we can write this state explicitly as
\be
|\Psi_{(2,N,N+1)} \rangle = \frac{1}{\sqrt{N+1}} \sum_{b=1}^N \left\{ \sqrt{\frac{N + 1 - b}{N}} |0 \rangle  \otimes  |b  \rangle  \otimes |b \rangle + \sqrt{\frac{b}{N}}  |1 \rangle  \otimes|b \rangle  \otimes |b+1 \rangle \right\} \; .
\ee
For general $K$, we can write the LME state as a tensor product
\be
|\Psi_{(2,N,(N+1))} \rangle  \otimes |\Psi_{(K,K)} \rangle
\ee
where $|\Psi_{(K,K)} \rangle$ is the Bell state
\be
|\Psi_{(K,K)} \rangle = \frac{1}{\sqrt{K}}  \sum_{i=1}^K  |i \rangle  \otimes  |i \rangle \; .
\ee
The tensor product state lives in a Hilbert space
\be
({\cal H}_2  \otimes {\cal H}_N  \otimes {\cal H}_{N+1})  \otimes ({\cal H}_K  \otimes {\cal H}_K) = {\cal H}_2  \otimes ({\cal H}_N  \otimes {\cal H}_K)  \otimes ({\cal H}_{N+1}  \otimes {\cal H}_K) = {\cal H}_2  \otimes {\cal H}_{NK}  \otimes {\cal H}_{(N+1)K}
\ee
as desired, where subscripts indicate dimensions.

\section{Locally maximally entangled states from representation theory}

In this section, we describe a special class of LME states that can be constructed using data coming from the representation theory of arbitrary finite and compact groups. We will see that most of the examples in the previous section can be understood in this way and that this method allows construction of explicit LME states in many other cases (or possibly even for all cases $(d_1, \dots, d_n)$ where LME states exist; see section 3.3).

We begin with
\begin{theorem}
Let $|\Psi \rangle$ be a pure state in a quantum system with Hilbert space
\[
{\cal H} = {\cal H}_A  \otimes {\cal H}_{\bar{A}} \, ,
\]
where $A$ and $\bar{A}$ label a subsystem and its complementary
subsystem. Let $H$ be a subgroup of $U(d_A) \times U(d_{\bar{A}})$
that acts irreducibly on ${\cal H}_A$ such that $|\Psi \rangle$ is
invariant under $H$ up to a phase (i.e. $h|\Psi \rangle = e^{i
\phi(h)}|\Psi \rangle$ for all $h \in H$). Then the reduced density
matrix $\rho_A$ is maximally mixed.
\end{theorem}
\begin{proof}
If $|\Psi \rangle$ is invariant under $H$ up to a phase, then the density matrix $\rho_A$ is invariant under the action of $H$. Letting $U_h$ be the representative of $h \in H$ in $U(d_A)$, we then have that for all $h$, $U_h \rho_A U_h^\dagger = \rho_A$, or $[U_h, \rho_A] = 0$. Since $h \to U_h$ is an irreducible representation of $H$, Schur's Lemma implies that $\rho_A$ is a multiple of the identity operator.\footnote{More directly, if $\rho_A$ is not maximally mixed, then $\rho_A^{ad} = \rho_A - \dfrac{1}{d_A} \identity$ does not vanish, and transforms in the adjoint representation of $U(d_A)$. We have seen that $U_h$ commutes with $\rho_A$ and therefore with $\rho^{ad}_A$ for all $h$. Acting on $A$, $U_h$ therefore does not mix subspaces with different eigenvalues for $\rho^{ad}_A$. Since $\rho^{ad}_A$ is nonvanishing and traceless, it must have at least two different eigenvalues, so there are proper invariant subspaces for the action of $H$, in contradiction with the assumption that $H$ acts irreducibly on $A$.}
\end{proof}
The theorem has an immediate application to multipart systems:
\begin{corollary}
\label{cor.multipart}
Consider a multipart quantum system with Hilbert space
\[
{\cal H} = {\cal H}_1  \otimes \cdots  \otimes {\cal H}_n \;
\]
upon which $\tilde{K} = U(d_1) \times \cdots \times U(d_n)$ acts. If a state $|\Psi \rangle \in {\cal H}$ is invariant up to a phase under a subgroup $H \subset \tilde{K}$, then the density matrix for each subsystem  ${\cal H}_{\alpha} =  \otimes_{i \in \alpha} {\cal H}_i$, upon which $H$ acts irreducibly is maximally mixed.
\end{corollary}
This gives a way to construct states of multipart systems whose elementary subsystems are all maximally mixed:
\begin{corollary} \label{cor.genstates}
Consider any group $H$ and any set of unitary irreducible representations $R_i$ of $H$ whose tensor product contains the trivial representation (i.e. an invariant vector).\footnote{The same result holds if the tensor product of $R_i$ contains any one-dimensional representation, since the associated one-dimensional invariant subspace gives a state $|\Psi \rangle$ invariant under $H$ up to a phase as required by Corollary \ref{cor.multipart}. However, this generalization does not produce any new examples: if $R$ is the one-dimensional representation in the tensor product and $R^*$ is the complex conjugate representation, then $(R_1 \otimes R^*,R_2, \dots R_n)$ will be a set of irreducible representations whose tensor product contains the trivial representation, and the corresponding invariant state will be the state $|\Psi \rangle$.} Given an explicit representation of such an invariant as
\be
\label{genstates}
|0 \rangle = \sum C_{a_1 \cdots a_n} |a_1 \rangle \otimes \cdots \otimes | a_n \rangle \; ,
\ee
where $C_{a_1 \cdots a_n}$ are group-theoretic coefficients describing how the trivial representation is embedded in the tensor product, the state $|0 \rangle$, considered as a quantum state in Hilbert space ${\cal H}_1 \otimes \cdots \otimes {\cal H}_n $ will have all elementary subsystems maximally mixed. Further, any composite subsystem ${\cal H}_{\alpha}$ for which the tensor product of representations $R_i$ with $i \in \alpha$ gives a single irreducible representation will also be maximally mixed.
\end{corollary}

\begin{remark}
In some cases, it is useful to consider a variant of our construction based on a representation
\be
h \longmapsto u_1(h)  \otimes \cdots  \otimes u_n(h) \in U(d_1\cdots d_n)
\ee
of a group $H$ for which the individual maps $h \mapsto u_i(h)\in H(d_i)$ give only projective representations of $H$ (i.e. where $u_i(h_1) u_i(h_2) = c_i(h_1,h_2)u_i(h_1 h_2)$ for some scalar function $c(h_1,h_2)$). If these projective representations are irreducible and if a state $|\Psi\rangle$ is invariant under the action of $H$ defined in this way, then essentially the same arguments as above show that $|\Psi\rangle$ will be LME. Here, we have the possibility that the group $H$ is abelian, since we can have irreducible projective representations of abelian groups with dimension greater than 1.

It is possible to understand this construction as a case of the general construction in Corollary~\ref{cor.genstates}, by finding a group $\hat{H}$ that is a central extension of $H$ and a representation of $\hat{H}$ in $U(d_1) \times \cdots \times U(d_n)$ that descends to the representation described above under the maps $\hat{H} \to H$, $(u_1,\dots, u_n) \to u_1 \otimes \cdots \otimes u_n$. For example, we can take $\hat{H}$ to be the subgroup of $K = U(d_1) \times \cdots \times U(d_n)$ generated by the elements $(u_1(h),\dots, u_n(h))$ for $h \in H$ and the representation to be the defining one.

Example 8 below will be of this type.
\end{remark}

\subsection{Examples}

Let us now describe various examples of LME states constructed in this way.

\subsubsection*{Example 1:  Bell states from dual representations}

For $n=2$, we can consider any group $H$ and any irreducible representation $R$ acting on a vector space of dimension $d$ and basis $e_i$ as
\be
\label{rep}
e_i \to R_{ij}[h] e_j \; .
\ee
Take $R^*$ to be the dual representation, acting on the dual vector space as
\be
\label{dualrep}
e^*_i \to  e^*_j R_{ji}[h^{-1}]
\ee
where $e^*_i$ are a basis for the dual vector space such that $(e^*_i, e_j) = \delta_{ij}$. The tensor product of the original vector space $V$ with the dual vector space $V^*$ is equivalent to the vector space of linear maps from $V \to V$. The tensor product of $R$ and $R^*$ acts on this space, and always includes a copy of the trivial representation, namely the identity map
\be
\label{identity}
1 \equiv \sum_i e_i  \otimes e_i^* \; .
\ee
We can check that this is invariant under the combined transformations (\ref{rep}) and (\ref{dualrep}). In the language of quantum states, the expression (\ref{identity}) corresponds to a normalized state
\be
|0 \rangle = \frac{1}{\sqrt{d}} \sum_i |i \rangle  \otimes |i \rangle \; ,
\ee
which is exactly the Bell state (\ref{Bell}). By choosing $H = \SU(2)$, we have examples for every positive integer $d$. We have seen above that this is the unique LME state for $n=2$ up to local unitary transformations.

\subsubsection*{Example 2: the GHZ state from $S_3$}

Moving on to tripartite systems, we first show how the GHZ state
(\ref{GHZ}) can be obtained from the construction in Corollary~\ref{cor.genstates}. We
need to find a group $H$ with two-dimensional irreducible
representations such that the product of three of these contains the
trivial representation. The smallest such group is $S_3$, the
permutation group on three elements; this has a unique two-dimensional
irreducible representation, and the tensor product of three of these
contains the trivial representation, so this should allow us to
construct an LME state in a three-qubit Hilbert space. Though we have
already proven that the GHZ state is the unique such state up to local
unitary transformations (section \ref{sec:Explicit23}), it may be
useful to demonstrate explicitly that we indeed find the GHZ state. We
can write the two-dimensional representation of $S_3$ explicitly via

\begin{align*}
(1)  = \left(\begin{matrix} 1&0\\0&1 \end{matrix} \right)&, (123)=\left( \begin{matrix} \omega^{-1}&0\\0&\omega  \end{matrix} \right), (132)=\left( \begin{matrix} \omega&0\\0&\omega^{-1}  \end{matrix} \right),\\
 (12)=\left( \begin{matrix} 0&1\\1&0 \end{matrix} \right)&,\, \, \,  (23)=\left( \begin{matrix} 0&\omega \\\omega^{-1}&0 \end{matrix} \right),\, \, \,  (13)=\left( \begin{matrix} 0&\omega^{-1} \\\omega&0 \end{matrix} \right)
\end{align*}

where $\omega = e^{\frac{2 \pi i}{3}}$. It is
then straightforward to check that the GHZ state (\ref{GHZ}) is
invariant under $S_3$ acting via the tensor product of three copies of
this representation:

\begin{align*}
(1)  \otimes (1)  \otimes (1)&,\quad  (123)  \otimes (123)  \otimes
(123),\quad  (132)   \otimes (132)
  \otimes (132) ,\quad \\
 (12)  \otimes
(12)  \otimes (12)&,\quad \quad  (13)  \otimes (13)  \otimes
(13),\quad (23)  \otimes (23)  \otimes (23).
\end{align*}

More generally, the GHZ state is invariant under the subgroup of $U(2)
\times U(2) \times U(2)$ consisting of elements of the form \be \left(
\ba{cc} e^{i \alpha} & 0 \cr 0 & e^{i \beta} \ea \right)  \otimes
\left( \ba{cc} e^{i \theta} & 0 \cr 0 & e^{i \chi} \ea \right)  \otimes
\left( \ba{cc} e^{-i \alpha - i \theta} & 0 \cr 0 & e^{-i \beta - i
\chi} \ea \right) \ee and \be \left( \ba{cc} 0 & e^{i \alpha} \cr e^{i
\beta} & 0 \ea \right)  \otimes \left( \ba{cc} 0 & e^{i \theta} \cr
e^{i \chi} & 0 \ea \right)  \otimes \left( \ba{cc} 0 & e^{-i \alpha - i
\theta} \cr e^{-i \beta - i \chi} & 0 \ea \right) \ee Thus, instead of
$S_3$, we could have started with this group or any subgroup of this
group acting irreducibly on all the factors.

\subsubsection*{Example 3: LME states from $\SU(2)$ 3j-symbols.}

As a more general example, when the dimensions satisfy triangle inequalities and sum to an odd number, the tensor product of $\SU(2)$ representations with dimensions $A$, $B$, and $C$ contains the trivial representation. The explicit representation of this trivial state is precisely described by the $3j$-symbols (closely related to Clebsch-Gordon coefficients):
\be
|0 \rangle =  \sum_{m_A,m_B,m_C} \left( \ba{ccc} \frac{A - 1}{2} & \frac{B - 1}{2} & \frac{C - 1}{2} \cr
m_A & m_B & m_C \ea \right) |m_A \rangle  \otimes |m_B \rangle  \otimes |m_C \rangle \; .
\ee
It is easy to verify that $\rho_A$, $\rho_B$, and $\rho_C$ are all proportional to the identity matrix using standard orthogonality relations for $3j$ symbols. For $A=2$, this gives examples with $(d_1,d_2,d_3) = (2,N,N+1)$. As we have seen above, for each $N$, there is a unique LME state with these dimensions up to local unitary transformations, so our representation theory construction gives the only example.

\subsubsection*{Example 4: states from product groups}

We can obtain more examples by considering product groups. For example, if a group $H$ has irreducible representations $\{R_1,R_2,R_3\}$ whose tensor product contains the identity and group $H'$ has irreducible representations $\{R_1',R_2',R_3'\}$ whose tensor product contains the identity, then group $H \times H'$ has representations $\{(R_1,R_1'),(R_2,R_2'),(R_3,R_3')\}$ whose tensor product contains the identity. This allows us to build examples with dimensions $(d_1 d_1',d_2 d_2', d_3 d_3')$. The states constructed in this way are simply tensor products
\be
|\Psi \rangle = |\psi \rangle  \otimes |\psi' \rangle
\ee
As an example, we can consider the group $\SU(2) \times \SU(2)$, and representations $(2,1)$, $(N,K)$, and $(N+1,K)$ (labeled by dimension) to construct examples with dimensions $(2,NK,(N+1)K)$. Again, we have seen above that there is a unique LME state up to local unitary transformations with these dimensions, so our construction provides the only example.

\subsubsection*{Example 5: states for $(d_1,d_2,d_3) = (2,N,N)$}

For dimensions $(2,d_2,d_3)$ we have seen how to construct all possible LME states using representation theory except for the cases $(2,N,N)$ with $N > 2$. We can construct examples with these dimensions so long as we can find a group $H_p$ and representations $\{R_2, R_p,\hat{R}_p\}$ giving a construction for the case $(2,p,p)$ with $p$ odd. In this case, we can use the group $H_p \times \SU(2)$ to construct examples with dimensions $(2,N,N)$ with $p$ a prime factor of $N$ by choosing representations $\{(R_2,1), (R_p,N/p),(\hat{R}_p,N/p)\}$ (we can use $H_2 = S_3$ as above). In appendix $B$ we show that a group with the desired properties is $H_p = UT(3,p) \rtimes \mathbb{Z}_2$ a certain semidirect product of $UT(3,p)$ (the group of upper triangular matrices with elements in $\mathbb{Z}_p$), and $\mathbb{Z}_2$. This is a finite group with $2p^3$ elements.

Thus, for every case with dimensions $(2,d_2,d_3)$ for which LME states exist, there is a construction based on representation theory that provides examples. In most cases, the LME state is unique up to local unitary transformations. For the cases $(2,N,N)$ with $N \ge 4$ we found a $2(N-3)$ real dimensional space of such states; in this case, we expect that the representation theory construction gives special states in this space which are invariant under a larger subgroup of $\SU(2) \times \SU(N) \times \SU(N)$.

\subsubsection*{Example 6: states from representation theory tensor networks}

For general $n$, we will now argue that $LME$ states obtained using our representation theory construction can always be represented by tensor networks with trivalent vertices, where the edges at a vertex are labelled by representations of our group whose tensor product contains the identity and tensors at each vertex correspond to the Clebsch-Gordon type coefficients describing how the trivial representation is constructed from the tensor product.

To see this, we note that for a given set of representations $R_1, \cdots, R_n$ associated to the elementary subsystems, we can determine the tensor product of representations recursively by first decomposing the tensor product of any pair of representations into a sum of irreducible representations, and then repeating this procedure (now with $n-1$ total representations) for each element in the sum. For each possible representation in the full tensor product, we can associate a tree graph with trivalent vertices, with the graph structure showing our choice for how to pair up representation and the edges labeled by the representations obtained in the intermediate steps. For our application, the final representation should be the trivial representation. To construct the corresponding LME state explicitly, we can interpret the graph as a tensor network, with vertices corresponding to Clebsch-Gordon coefficients that tell us how the representation for each outgoing edge is obtained from the tensor product of the two representations associated with incoming edges at that vertex. This is shown in figure \ref{fig:tensorBoth}a.

We can represent the state using an equivalent tensor network where all the trivalent vertices have incoming legs (i.e. a PEPS network) making use of a simple representation theory fact: for any representations $R_1$ and $R_2$ whose tensor product contains $R_3$, the tensor product of $R_1$, $R_2$, and the dual representation $R_3^*$ contains the trivial representation. Thus, we can replace the Clebsch-Gordon coefficient for a vertex with two incoming and one outgoing leg with the related tensor constructing the trivial representation from the $R_1$,$R_2$,$R_3^*$ and a tensor constructing the trivial representation from $R_3$ and $R_3^*$. For the case of $\SU(2)$, this means that we are using $3j$-symbols rather than Clebsch-Gordon coefficients. This is indicated in figure \ref{fig:tensorBoth}b.

\begin{figure}
\centering
\includegraphics[width=0.9\textwidth]{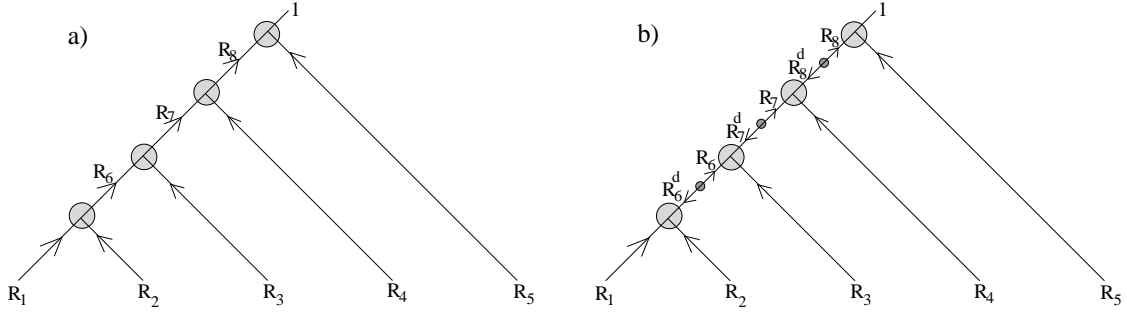}
\caption{Tensor networks describing LME states constructed from irreducible representations of a group. On the left, trivalent vertices correspond to Clebsch-Gordon type tensors describing how the representation corresponding to the outgoing leg arises from the tensor product of representations corresponding to the incoming legs. On the right, the trivalent vertices correspond to tensors constructing the trivial representation from the product of the three incoming representations. Small bivalent vertices represent tensors giving the trivial representation from the product of a representation and its dual.}
\label{fig:tensorBoth}
\end{figure}

For a given set of representations $(R_1, \cdots, R_n)$, we will often have a number of copies of the trivial representation in the tensor product. These will correspond to different $LME$ states associated with graphs whose internal edges are labeled by different representations.

We need only consider one possible graph structure to represent any invariant state, since the graph structure corresponds to our choice for which order to combine representations. An LME state constructed using some other graph structure will be a linear combination of states constructed using the original graph structure.

As an example, consider the case of four qubits, and choose $H = \SU(2)$. The product of four spin half representations contains two copies of the trivial representation, so we get two different states, corresponding to an intermediate representation of spin 0 or spin 1. We can represent these by the tensor networks in figure \ref{qubit4simple}, or explicitly as\footnote{Equivalently, we could have used $3j$-symbols.}
\bea
|\Psi_0 \rangle &=& \epsilon_{\alpha \beta} \epsilon_{\sigma \rho} | \alpha \rangle  \otimes |\beta \rangle  \otimes |\sigma \rangle  \otimes |\rho \rangle \cr
|\Psi_1 \rangle &=& \sigma^i_{\alpha \beta} \sigma^i_{\sigma \rho} | \alpha \rangle  \otimes |\beta \rangle  \otimes |\sigma \rangle  \otimes |\rho \rangle
\eea
where $\epsilon_{\alpha \beta}$ is the antisymmetric tensor and $\sigma^i_{\alpha \beta}$ are Pauli matrices. The similar states defined using different graph structures are linear combinations of these.

\begin{figure}
\centering
\includegraphics[width=0.3\textwidth]{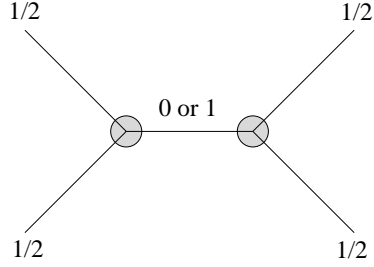}
\caption{LME states of four qubits corresponding to the two ways of combining four spin half representations of $\SU(2)$ to obtain a singlet.}
\label{qubit4simple}
\end{figure}

\subsubsection*{Example 7: LME states with maximally mixed composite subsystems}

In some cases, we may wish to construct states for which some composite subsystems are also maximally mixed. An example are $k$-uniform states of $N$ qudits, where all subsystems of less than or equal to $k$ qudits are maximally mixed (see, for example \cite{Scott, GZ}). According to Corollary~\ref{cor.genstates}, we can achieve this by ensuring that the product of representations corresponding to the elementary subsystems in the composite subsystem gives a single irreducible representation.

As a simple example, suppose we wish to construct states of a cyclic spin chain of size $kL$ for which subsystems of less than or equal to $L$ neighboring spins are all maximally entangled. To do this, we can choose $H = \SU(2)^L$, and assign representations $(\dfrac{1}{2}, 0,\dots)$, $(0, \dfrac{1}{2}, 0,\dots)$, and so forth around the chain. Then the tensor product of representations corresponding to groups of $L$ or fewer neighboring spins will give a single irreducible representation (e.g. $(\dfrac{1}{2},\dfrac{1}{2}, 0, \dots, 0)$ for the first pair of spins). For size $2L$, the tensor product of all the representations contains a single copy of the trivial representation, so this construction gives a single LME state. It turns out that this is simply the state obtained by making Bell pairs out of opposite spins in the ring. For larger systems, we have more possible states, including states with no Bell pairs. The states for the case $L=2$ and $k=4$ for which our construction gives four possible states is shown in figure \ref{group2}.

\begin{figure}
\centering
\includegraphics[width=0.5\textwidth]{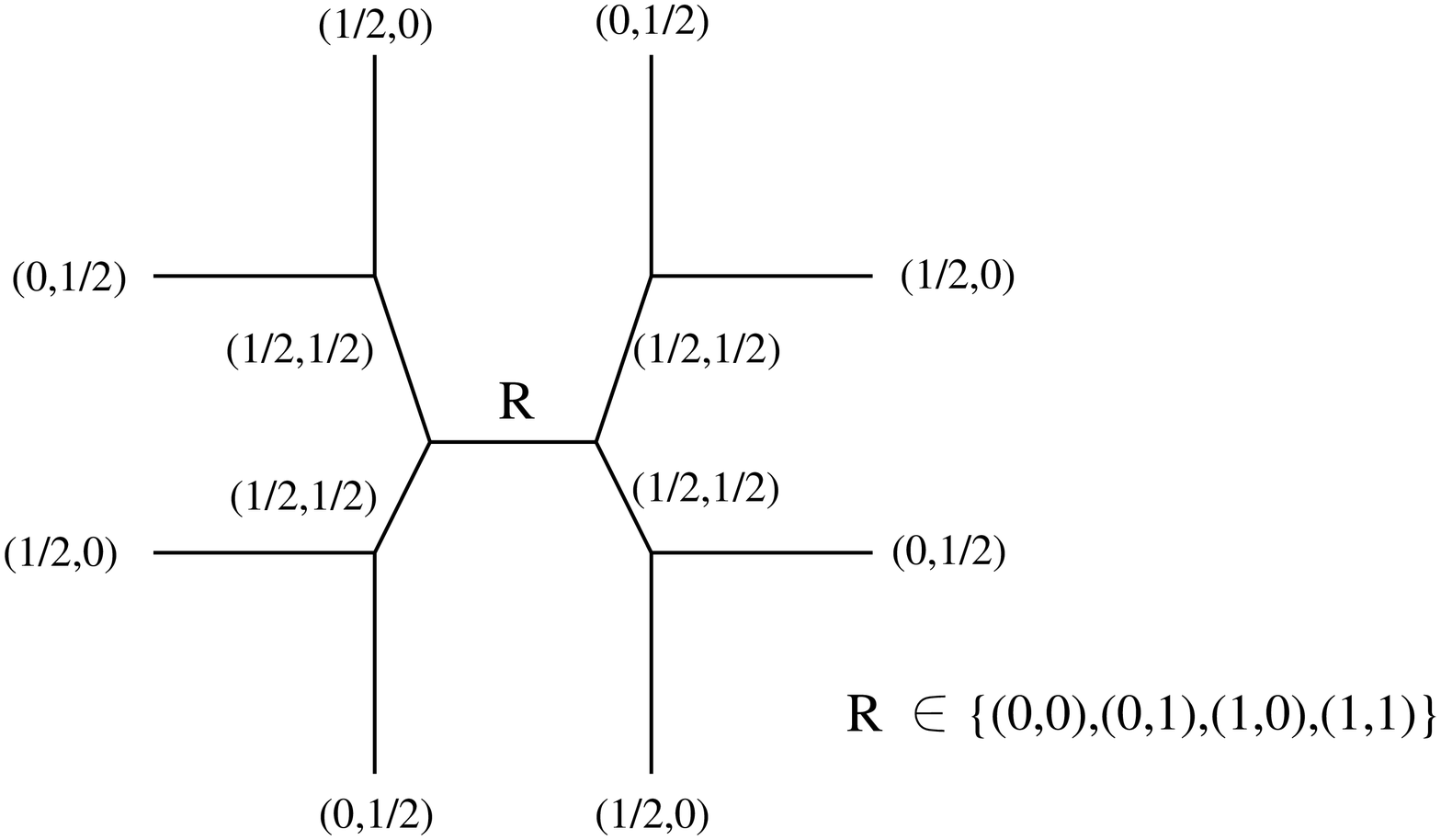}
\caption{Tensor network depiction of an eight qubit state with each spin and each nearest neighbor pair of spins in a maximally mixed state. }
\label{group2}
\end{figure}

\subsubsection*{Example 8: absolutely maximally entangled states and perfect tensors}

We can use this general approach to construct {\it absolutely
maximally entangled} (AME) states (also referred to as {\it maximal multipartite entangled states} (MMES) \cite{FFPP}), for which all subsystems with
dimension less than or equal to the dimension of the complementary
subsystem (call these ``small'' subsets) are maximally
mixed. Some previous discussions of these states, including complete existence criteria for qubit systems may be found in \cite{GB, HS, Scott, Helwig:2012nha, FFPP, HGS}.

According to Corollary~\ref{cor.genstates}, we can obtain a maximally entangled state if we find a
group $H$ with representations $R_i$ whose tensor product contains the
identity representation and for which the tensor product of any small
subset of representations is irreducible.

As an example, consider the six-qubit state described in Appendix A.1 of \cite{Pastawski:2015qua}. Here, we have a representation of $(\mathbb{Z}_2)^6$ to $U(d_1 \cdots d_n)$ of the form described in Remark 3.4, where the six generators are represented as
\bea
S_1 &=& X  \otimes Z  \otimes Z \otimes X \otimes I \otimes I \cr
S_2 &=& I \otimes X \otimes Z \otimes Z \otimes X \otimes I \cr
S_3 &=& X \otimes I \otimes X \otimes Z \otimes Z \otimes I, \cr
S_4 &=& Z \otimes X \otimes I \otimes X \otimes Z \otimes I \cr
S_5 &=& X \otimes X \otimes X \otimes X \otimes X \otimes X \cr
S_6 &=& Z \otimes Z \otimes Z \otimes Z \otimes Z \otimes Z
\eea
with
\be
I= \left(\ba{cc}1&0 \cr 0&1\ea\right) \qquad X = \left(\ba{cc}0&1 \cr 1&0\ea\right) \qquad Z = \left(\ba{cc}1&0 \cr 0&-1\ea\right) \; .
\ee
In this case, the corresponding representations of $(\mathbb{Z}_2)^6$ on the individual subsystems are projective. Their tensor product includes the trivial representation: it is not hard to check directly that there is a state with  $ |\Psi \rangle = |\Psi \rangle$. Also, it is not hard to check that the individual representations are irreducible and the product of any two or three of these is also irreducible.\footnote{For example, in each case, we can verify that any matrix commuting with the representatives of each generator of $(\mathbb{Z}_2)^6$ is a multiple of the identity matrix.} Thus, all subsystems of one, two, or three qubits are maximally entangled. To understand this example in the language of Corollary \ref{cor.genstates}, we can choose a cover group $\hat{H}$ to be the subgroup of $K$ generated by elements $S_i'$ with $S_1' = (X,Z,Z,X,I,I)$ and so forth. This is a central extension of $H$ of order $64 \times 32 = 2048$ with 32 central elements of the form $(\pm I, \pm I, \pm I, \pm I, \pm I, \pm I)$ with an even number of minus signs.

\subsection{Quantum systems with only LME states}

An interesting point (suggested by Eliot Hijano) is that we can describe quantum mechanical systems for which all physical states are LME. Given any multipartite system with a Hamiltonian ${\mathcal H}$ invariant under the global symmetry group $H$ acting irreducibly in each elementary subsystem (and possibly in some composite subsystems), we simply gauge the group $H$. This restricts us to physical states which are invariant under $H$, so by our Corollary~\ref{cor.genstates}, any subsystem upon which $H$ acts irreducibly will be maximally mixed. In this case, local observables cannot be used to distinguish the possible physical states.

For these theories, the dimension of the Hilbert space of physical states is exactly the number of trivial representations appearing in the tensor product of representations $R_i$ associated to the elementary subsystems. Given a particular graph structure for the associated tensor network (as in figure~\ref{group2}), we can choose basis elements corresponding to the different ways of labeling the edges with representations such that the tensor product of three representations at each vertex includes the identity. In some cases, when the tensor product of two representations at a vertex contains multiple copies of the dual of the third representation, there will be multiple independent ways to couple the three representations to a trivial representation. In this case, we need to add an additional discrete label on the vertex to indicate which structure we are choosing.

Alternatively, we can work in the ungauged theory and add a term to the Hamiltonian which associates a large energy to states which are not in the trivial representation of $H$. For example, with $H = \SU(2)$, we can add $E_0 J^2$ for large $E_0$. If there are no other terms in the Hamiltonian, the model will have a ground state degeneracy equal to the multiplicity of the trivial representation in the tensor product of representations $R_i$ associated with the elementary subsystems. Since these systems  have an energy gap to other states for which the local subsystems are not maximally mixed, quantum systems of this type (if they can be realized) may have applications for quantum error correction or the experimental realization of robust qubits/qudits.

\subsection{Implications for representation theory}

The construction of Corollary \ref{cor.genstates} combined with our Theorem \ref{thm.R} on the existence of LME states leads to an interesting implication for the representation theory of finite and/or compact groups:
\begin{proposition}
For dimensions $(d_1, \dots, d_n)$, there can exist a (finite or compact) group $G$ with unitary irreducible representations $R_1$, $R_2$, \dots $R_n$  of dimensions $(d_1, \dots, d_n)$ whose tensor product contains the trivial representation only if the conditions in Theorem \ref{thm.R} for the existence of LME states are met.
\end{proposition}

In the special case $n=3$, the tensor product of three irreducible representations $R_1$,$R_2$, and $R_3$ contains the trivial representation if and only if the tensor product of any two of the representations contains the dual of the third. Thus, we have
\begin{corollary}
For dimensions $(d_1, d_2, d_3)$, there can exist a (finite or compact) group $G$ with unitary irreducible representations of dimension $d_1$ and $d_2$ whose tensor product contains an irreducible representation of dimension $d_3$ only if
\be
d_1 d_2 d_3 > N(d_1^2, d_2^2, d_3^2)
\ee
where $N$ is defined in (\ref{defN}).
\end{corollary}

It is interesting to ask whether the necessary conditions in Proposition 3.5 and Corollary 3.6 might also be sufficient.

For $n=3$, the first five examples from section 3.1 show that sufficiency holds for $n=2$ and for $n=3$ with $d_1 = 2$. Using the GAP database of finite groups, it is possible to search for groups with representations of particular dimensions and print the character tables to check whether the tensor product of two representation contains a third. Using this approach, we have also checked that for each case of the form $(d_1,d_2,d_3) = (3,3,d_3)$ admitting LME states (i.e. $3 \le d_3 \le 9$), we can find a group and irreducible representations of dimensions $(3,3,d_3)$ whose tensor product contains the trivial representation. However, this approach runs out of steam quickly due to the finite size of the database.

For $n \ge 4$ it is straightforward to see that sufficiency can't hold. To see this, note that if there exist representations with dimensions $\{d_1,d_2,d_3,d_4\}$ whose tensor product contains the identity, then the tensor product of any two of the representations must contain a representation whose dimension is the same as a representation in the tensor product of the other two representations. But this cannot work for dimensions $(2,2,2,7)$ (for which LME states exist), since $2 \times 2$ can contain no representation of dimension larger than 4, while from Corollary 3.6, we find that $2 \times 7$ can contain no representation of dimension smaller than 6.

\section{Background: The geometry of LME states}

In this review section, we describe how the set of LME states has
two natural geometrical formulations that turn out to be equivalent
to each other. The first is related to symplectic geometry and
the other is related to algebraic geometry and geometric invariant
theory. Physically, these two perspectives relate to two seemingly
different classification problems for quantum states. This material has
been discussed previously in the literature; see for example~\cite{brylinsky1},
\cite{brylinsky2}, \cite{Kly02}, \cite[\S3]{Kly07}, \cite{gw10}, \cite{SOK12}, \cite[\S4]{wallach},
\cite{Walter} and \cite{SOK14}.
Additional reviews on the classification of multipart entanglement and the geometry
of multipart Hilbert spaces can be found in \cite{ES} and \cite{BZ}, respectively.
Readers already familiar with these geometrical
preliminaries can skip directly to section 5.

\subsection{Geometry of the space of states}

We begin with the full space of (unnormalized) states. For dimensions $(d_1,\dots, d_n)$, the Hilbert space ${\cal H} = {\cal H}_1  \otimes \cdots  \otimes {\cal H}_n$ is the complex vector space  $\mathbb{C}^{D}$ with $D = d_1 d_2 \cdots d_n$ whose coordinates can be taken as the coefficients $\psi_{i_1 \dots i_n}$ defining the state. The inequivalent physical states can be represented as equivalence classes of vectors with unit norm with two states identified if they are multiplicatively related by a phase. Equivalently, we can work with unnormalized states, omitting the zero vector and identifying states related via multiplication by a complex scalar. This defines the complex projective space $\bbP({\cal H}) = \mathbb{C P}^{D-1}$.

\subsubsection{Entanglement structure and the action of $K = \SU(d_1) \times \cdots \times \SU(d_n)$}

It is an interesting question to characterize the possible
entanglement structures that such states can have. By ``entanglement
structure'' we mean properties of a state that are unaffected by local
unitary transformations; that is, unitary transformations that act
independently on each tensor factor of the Hilbert space. These
transformations correspond to changes of basis for the individual
subsystems. Mathematically, the change of basis operations correspond
to the group $\tilde{K} = U(d_1) \times \cdots \times U(d_n)$ acting
on ${\cal H}$. Without loss of generality, we may consider the smaller
group $K = \SU(d_{1})\times \dotsb \times \SU(d_{n})$ since the action
of $\tilde{K}$ and $K$ on the projective space $\bbP({\cal H}) = \mathbb{CP}^{D-1}$ have the same orbits.

Each state in $\bbP({\cal H})$ will be on some orbit of $K$. The space of these orbits then represents the space of possible entanglement structures. To parameterize the space of these orbits, we can define a set of coordinates which are polynomials in $\psi$ and $\psi^\dagger$ invariant under the action of $K$ (for example, the traces of powers of reduced density matrices for various subsystems). The LME states that are the focus of this paper correspond to specific orbits of $K$; one way to characterize these in terms of $K$-invariants is to say that $\tr(\rho_i^2) = 1/d_i$ for each $i$.

\subsubsection{SLOCC orbits and the action of $G = \SL(d_1,\mathbb{C}) \times \cdots \times \SL(d_n,\mathbb{C})$}

Sometimes, we may be interested in a coarser classification of entanglement structure. Two states are said to be equivalent under the set of ``stochastic local operations and classical communication'' (SLOCC equivalent) if we can move from one to the other by performing reversible quantum operations on the individual subsystems.

Mathematically, the set of allowed SLOCC operations corresponds to the group $G = \SL(d_1, \mathbb{C}) \times \cdots \times \SL(d_n, \mathbb{C})$ of arbitrary invertible local transformations acting on $\bbP({\cal H})$ \cite{DVC}. Alternatively, we can consider $\tilde{G} =  \GL(d_1, \mathbb{C}) \times \cdots \times \GL(d_n, \mathbb{C})$ acting on ${\cal H}$.

As in the above discussion, each state will be on some orbit of $G$. Classifying states up to $SLOCC$ equivalence means understanding how $\bbP({\cal H})$ decomposes into orbits of $G$. As we discuss further below, we can choose a set of $G$-invariant polynomials (some subset of the entanglement invariants discussed above) as coordinates on the space of these orbits.

We will see below that the classification of SLOCC equivalence classes is intimately related to the classification of equivalence classes of LME states up to local unitaries.

\subsection{$\LME/K$ as a symplectic manifold}

The space of quantum states (either ${\cal H}$ or $\bbP({\cal H})$) has a natural symplectic structure (i.e. the structure of a phase space in Hamiltonian classical mechanics), associated with the symplectic form
\be
\omega = i h_{ij} dz^{i} \wedge  d\bar{z}^j =  \frac{i}{2} \frac{\partial}{\partial z^i} \frac{\partial}{\partial \bar{z}^j} \log (|z|^2) dz^{i} \wedge  d\bar{z}^j
\ee
on ${\cal H} = \mathbb{C}^D$ or the naturally associated pullback to $\bbP({\cal H})$. The latter is known as the Fubini-Study form. The associated metric
\be
ds^2 = h_{ij} dz^i d \bar{z}^j
\ee
is (up to an overall normalization) the same as the Bures metric that gives a natural measure of distance between quantum states.

The symplectic form is invariant under the action of $K = \SU(d_1) \times \dots \times \SU(d_n)$ on the phase space. Each infinitesimal transformation  in this group, corresponding to some element $k \in Lie(K)$ of the associated Lie algebra, may be associated with a vector field $X_k$ on phase space indicating the infinitesimal form of the transformation. Via the symplectic form, such a vector field can be associated with a real Hamiltonian function $H_k$ on $\bbP({\cal H})$ as
\be
d H_k = {X_k} \cdot \omega \; .
\ee
The map from the symmetry generator $k$ to the function $H_k$ on phase space is precisely the usual map between symmetries and conserved quantities in Hamiltonian mechanics. Mathematically, this is referred to as a {\it comoment map}.

The same relation can also be expressed as a {\it moment map},
\be
\mu : \bbP({\cal H}) \to Lie(K)^*
\ee
which associates to every state $x$ in $\bbP({\cal H})$ the function ($k \to H_k(x)$) from $Lie(K)$ to real numbers (i.e. a vector in the dual space $Lie(K)^*$).

We now show that the LME states are exactly the subset $\mu^{-1}(0) \subset \bbP({\cal H})$. To see this, we note that $\mu(x) = 0$ if and only if $H_k(x)$ vanishes for some basis of elements $k \in Lie(K)$. For our group $K$, these basis elements can be chosen to take the form
\be
k = \identity  \otimes \cdots \otimes k_i \otimes \cdots  \otimes \identity
\ee
where $k_i$ is a generator of $\SL(d_i,\mathbb{C})$ (i.e. a traceless $d_i \times d_i$ matrix). For the point $x \in \bbP({\cal H})$ corresponding to a state $|\Psi \rangle$, explicit calculation shows that
\be
H_k(x) = \langle \Psi| \identity  \otimes \cdots \otimes k_i \otimes \cdots  \otimes \identity | \Psi \rangle = \tr (\rho_i k_i) \; .
\ee
Thus, we have $H_k(x)$ vanishing for all $k$ (so that $\mu(x)=0$) if and only if the trace of each reduced density matrix $\rho_i$ multiplied by any traceless matrix equals zero. This will be true if and only if each reduced density matrix is a multiple of the identity matrix i.e. the subsystem is maximally mixed.

In summary, the space of locally maximally entangled states is the same as the inverse image of $0$ under the moment map associated with $K = \SU(d_1) \times \cdots \times \SU(d_n)$. In the classical mechanics language, it is the set of points on phase space where all conserved quantities associated with infinitesimal local unitary transformations vanish. The space of equivalence classes of these states up to local unitary transformations is then the quotient $\mu^{-1}(0)/K$. A general result in symplectic geometry is that a quotient space defined in this way (known as the {\it symplectic quotient}) is also a (possibly singular) symplectic manifold.

\subsection{$\LME/K$ as a complex manifold}

By a remarkable duality known as the Kempf-Ness theorem, the
symplectic quotient $\mu^{-1}(0)/K$ that we have identified with
$\LME/K$ is equivalent to another quotient $\bbP({\cal H})\GITquot G$,
known as the {\it geometric invariant theory} (GIT)
quotient of the full space of states $\bbP({\cal H})$ by the larger group
$G = \SL(d_1, \mathbb{C}) \times \cdots \times \SL(d_n, \mathbb{C})$
that is the complexification of $K$.\footnote{For a more complete discussion of geometric
invariant theory, symplectic geometry, and the Kempf-Ness theorem, see
\cite{mfk}, or see \cite{Hoskins} for a pedagogical introduction.} This group is the set of local
invertible transformations (with unit determinant), known in quantum
information literature as the group of transformations under SLOCC
(stochastic local operations and classical communication).

To explain the equivalence, consider the action of $G$ on the full vector space ${\cal H} = \mathbb{C}^{d_1 \cdots d_n}$ before normalization. The space ${\cal H}$ decomposes into orbits of $G$, each of which contains vectors of different norm (see figure \ref{fig:cone}).

It is useful to distinguish between {\it unstable} orbits for which the infimum of the norm is zero (i.e. the orbit contains states of arbitrarily small norm), and the remaining {\it semistable} orbits for which the infimum of the norm is positive. A subset of semistable orbits known as the {\it polystable} orbits contain states of minimum norm on which the infimum is achieved; the polystable orbits are those that are topologically closed.

It turns out that these states of minimum norm on polystable $G$ orbits are precisely the LME states.\footnote{Here, we are including states with arbitrary normalization.}

To see this, we will first show that a state $|\Psi \rangle \in {\cal H}$ is LME if and only if the function
\be
\label{KNfunction}
g \in G \to |g |\Psi \rangle|^2
\ee
has an extremum at $g=1$. The vanishing of the first order variation of the norm-squared function at $g=1$ is equivalent to
\be
\langle \Psi| k |\Psi \rangle = 0
\ee
where $k$ is an element of the Lie algebra of $K = \SU(d_1,\mathbb{C}) \times \cdots \times \SU(d_n, \mathbb{C})$, the maximal compact subgroup of $G$. This is precisely the same condition we obtained from the vanishing of the moment map, so by the arguments above, the vectors satisfying this condition are exactly the maximally mixed ones.

By looking at the second derivative of (\ref{KNfunction}), it can be shown that any extremal vector, sometimes referred to as a {\it critical state}, must be a local minimum.  Further, this minimum is unique on the $G$-orbit up to the action of $K$ (which leaves the norm invariant).

Thus, a state is LME if and only if it is of minimum norm on some orbit, and a single $G$-orbit contains at most one $K$-orbit of LME states. It follows that the space $\LME/K$ of LME states up to local unitary transformations is equivalent to the space of polystable $G$-orbits.

A further equivalence comes by noting that each semistable orbit contains a unique polystable orbit in its closure. Thus, it is natural to define an equivalence relation by which two semistable orbits are in the same equivalence class if they have the same polystable orbit in their closure. In this way, the space of polystable orbits is equivalent to the space of equivalence classes of semistable orbits.

This space of equivalence classes of semistable orbits defines the {\it geometric invariant theory (GIT) quotient} of ${\cal H}$ by $G$, denoted ${\cal H}\GITquot G$; to summarize, starting from the full space ${\cal H}$ we define ${\cal H}\GITquot G$ by throwing out the unstable points, taking the quotient by $G$ and identifying orbits via the equivalence relation. This space has nicer geometrical properties than the naive topological quotient ${\cal H}/G$. The direct quotient ${\cal H}/G$ is typically not even Hausdorff,\footnote{Recall that a Hausdorff space is one where any two distinct points $x$ and $y$ are contained in some disjoint neighborhoods $U_x$ and $U_y$. When this fails, we can have unusual features such as convergent sequences that do not have a unique limit.} while the GIT quotient ${\cal H}\GITquot G$ is a (possibly singular) complex manifold with a nice algebraic characterization that we describe below.

Starting from ${\cal H}\GITquot G$, we can identify orbits related by complex multiplication to define the quotient $\bbP({\cal H})\GITquot G$.\footnote{Here, it is important that scalar multiplication by complex numbers commutes with the action of $G$, so orbits related by scalar multiplication have the same stability properties.} From the discussion above, it follows that inequivalent LME states are in one-to-one correspondence with points in $\bbP({\cal H})\GITquot G$. We therefore have two different geometrical characterizations of the set $\LME$ of locally maximally entangled states up to local unitary equivalence.
\be
\LME/K \cong \mu^{-1}(0)/K \cong \bbP({\cal H})\GITquot G \; .
\ee
The latter equivalence here is guaranteed by a result known as the Kempf-Ness theorem.

A consequence of the equivalence between $\bbP({\cal H})\GITquot G$ and
$\mu^{-1}(0)/K$ is that this space has both complex and symplectic
structure. These structures are compatible, so $\LME/K$ is a
K\"ahler manifold. It can be described explicitly in terms of the
complex coordinates $\psi_{i_1 \dots i_n}$ defining the state by
giving a finite set of holomorphic polynomials $P_{\alpha }$ in the
variables $\psi_{i_{1}\dotsb i_{n}}$, which are invariant under $G$,
together with all polynomial relations $R_n(P_\alpha ) = 0$ satisfied
by the $P_{\alpha }$. We describe this construction in more detail below.

\subsection{Gradient flow to LME states}

We have seen that each orbit of $G$ either contains the zero vector in its closure or contains a unique $K$-orbit of LME states in its closure. In this subsection, we recall that there is a natural vector field on ${\cal H}$ for which the associated flow takes us from any point $\psi$ in ${\cal H}$ along a path through the orbit $G \psi$ to either the origin (in the unstable case) or to an LME state (in the semistable case). Via this flow, each semistable point in ${\cal H}$ may then be associated with a specific LME state.

First, recall that the comoment map takes a point in the Lie algebra $Lie(K)$ for $K = \SU(d_1) \times \cdots \times \SU(d_n)$ to a Hamiltonian function $H_k = \psi^\dagger k \psi$ on ${\cal H}$. Using the natural inner product on the Lie algebra, we can define an orthonormal basis
\be
k_{a,i} = \identity  \otimes \cdots  \otimes T_{d_i}^a  \otimes \cdots  \otimes \identity \; ,
\ee
of Lie algebra elements, where $T^a_{d_i}$ are generators of $\SU(d_i)$ normalized so that
\be
\tr(T^a_{d_i} T^b_{d_i}) = \frac{1}{2} \delta^{a b} \; .
\ee
We can then define the single function $M : {\cal H} \to \mathbb{R}$ as
\be
\label{defM}
M = \sum_{a,i} H^2_{k_{a,i}} = \sum_{a,i} (\psi^\dagger k_{a,i} \psi)^2 = \sum_{a,i} \tr^2(\rho_i T^a_{d_i})   \; ,
\ee
known as the ``square of the moment map''. We have seen that for nonzero $\psi$, $H_k(\psi)$ vanishes for all $k$ if and only if $\psi \in \LME$. Thus $M$ is minimized on the subset $\LME \in {\cal H}$. If we now define $\vec{v} = -\nabla M$, we have a vector field which points in the direction of steepest descent for the function $M$. We will see that the flow defined by this vector field takes us from any semistable state $\psi$ to an LME state in the $G$-orbit of $\psi$, and from any unstable state to the zero vector.

To proceed, let us derive a more explicit form for $M$ and for the associated gradient field. Using the fact that generators of $\SU(d_i)$ form a basis of all traceless Hermitian matrices, we can derive the completeness relation
\be
\sum_a (T^a_{d_i})_{jk} (T^a_{d_i})_{lm} = \frac{1}{2} (\delta_{jm} \delta_{kl} - \frac{1}{d_i} \delta_{jk} \delta_{lm}) \; .
\ee
This can be used to simplify (\ref{defM}) as
\be
M = \frac{1}{2} \sum_{i} \left(\tr(\rho_i^2) - \frac{1}{d_i} \tr^2(\rho_i)\right) \; .
\ee
Here, we are working with unnormalized states, so $\tr(\rho_i)$ can take any positive value. We see that the function $M$ is independent of our choice of basis. Varying $M$ with respect to $\psi^\dagger$ to determine the gradient, we find that the associated flow is
\be
\label{gradflow}
\frac{d}{d\lambda} \psi_{j_1 \cdots j_n}  = - \sum_i  (\hat{\rho}_i)_{j_i} {}^{k} \psi_{j_1 \cdots k \cdots j_n} \qquad \qquad \hat{\rho}_i \equiv \rho_i - \frac{1}{d_i} \tr(\rho_i) \identity \; .
\ee
The gradient vector on the right side vanishes if and only if $\rho_i = \dfrac{1}{d_i} \tr(\rho_i) \identity$, which is possible if and only if $\psi \in \LME$ or $\psi = 0$, so these are the only allowed endpoints for the flow.

It is not hard to show that for a point $\psi \in {\cal H}$, infinitesimal flow along the gradient direction coincides with the action of an infinitesimal element in $G$ defined by maximizing the rate of decrease of the function (\ref{KNfunction}) over all Lie algebra elements in $Lie(G)$ of some fixed norm.\footnote{To see this, we minimize the $\lambda$ derivative of $|(1 + \lambda k) \cdot \psi|^2$ at $\lambda = 0$ over the possible $k \in Lie(G)$ subject to the constraints that $\tr(k k^\dagger) = 1$ and $\tr(k) = 0$. The flow associated with the resulting generator gives precisely the result (\ref{gradflow}).}  Thus, the gradient flow remains within an orbit of $G$, and acts to decrease the norm of the state. For semistable $\psi$, the function (\ref{KNfunction}) is bounded below by a positive value on the orbit $G \cdot \psi$, so the state reached in the limit along the flow from $\psi$ has positive norm and must be a LME state. For unstable $\psi$, there are no LME states in the closure of $G \cdot \psi$ so the flow from $\psi$ must approach the zero vector in the limit.

These ideas have been utilized recently in the quantum information literature on multipart entanglement
\cite{SOK14, WDGC, MS} in order to find and characterize important states that are critical points for certain measures of entanglement. These include LME states but also certain special states on the unstable orbits. A discrete version of the gradient flow with similar properties was described in \cite{VDD} as an algorithm to associate an LME state as a ``normal form'' for a general state in the Hilbert space. See also \cite{BGOWW} for a recent discussion.

\subsection{Algebraic characterization of $\bbP({\cal H})\GITquot G$}

Let us now discuss the algebraic characterization of the GIT quotient $\bbP({\cal H})\GITquot G$. The points in ${\cal H} = \mathbb{C}^{d_1 d_2 \cdots d_n}$ are labeled by complex numbers $\psi_{i_{1} \cdots i_n}$ such that the associated (unnormalized) quantum state is
\be
|\Psi \rangle = \sum_{\vec{i}} \psi_{i_1 \cdots i_n} |i_1 \rangle  \otimes \cdots  \otimes |i_n \rangle
\ee
Under $G = \SL(d_1,\mathbb{C}) \times \cdots \times \SL(d_n,\mathbb{C})$, $\psi_{i_{1} \cdots i_n}$ transforms as
\be
\psi_{i_1 \cdots i_n} \to M^1_{i_1} {}^{j_1} \cdots M^n_{i_n} {}^{j_n} \psi_{j_1 \cdots j_n}
\ee
where $M^i$ is a $d_i \times d_i$ matrix with unit determinant. Certain polynomials in these coordinates are invariant under the action of $G$. For example, in the $d_1 = d_2 = 2$ case, we have
\be
\epsilon^{ij} \epsilon^{kl} \psi_{ij} \psi_{kl} = \psi_{11} \psi_{22} - \psi_{12} \psi_{21} \; .
\ee
which is invariant since it is the determinant of
\be
Z = \left( \ba{cc} \psi_{11} & \psi_{12} \cr \psi_{21} & \psi_{22} \ea \right)
\ee
and we have
\be
\det(Z) \to \det(M^1 Z (M^2)^T) = \det(M^1) \det(Z) \det(M^2) = \det(Z) \; .
\ee
More generally, we can build invariants by taking some number  of copies of $\psi_{i_1 \cdots i_n}$ (this must be a multiple of ${\rm lcm}(d_1, \dots, d_n)$) and contracting the set of $k$th indices on all the copies in some way using invariant tensors $\epsilon^{i_1 \cdots i_{d_k}}$. An important result of Hilbert (see \cite{mfk,dc}) is that the ring of all such invariants is finitely generated over $\mathbb C$. That is, there is a finite number of polynomial invariants $f_1, \dots, f_N$, which can be taken to be homogeneous of positive degree, such that that every other invariant can be expressed as a polynomial in $f_1, \dots, f_N$ with complex coefficients. Let us denote the degree of $f_i$ by $k_i$. In general, we can have relations among the generators $f_1, \dots, f_N$.

For a point $x$ in the projective space $\bbP({\cal H})$, different representatives in ${\cal H}$ will have different values for the invariant polynomials. But for $\psi \to \lambda \psi$, our basis of polynomials transform as
\be
(f_1,\dots,f_N) \to  (\lambda^{k_1} f_1,\dots, \lambda^{k_N} f_N) \; .
\ee
Thus, to any orbit of $G$ in $\bbP({\cal H})$, we can associate an $n$-tuple of complex numbers defined up to the equivalence relation
\be
(f_1,\dots,f_N) \sim  (\lambda^{k_1} f_1,\dots, \lambda^{k_N} f_N) \; .
\ee
This defines the {\it  weighted projective space} $\mathbb{C P}(k_1,\dots,k_N)$. We will denote the equivalence classes by $(f_1:\dots:f_N)$. Taking into account the algebraic relations between the generators, the space of values for these invariant polynomials will be an algebraic variety in the weighted projective space (i.e. a subspace defined by some polynomial equations).

Basic results in geometric invariant theory tell us that the geometry of the quotient $\bbP({\cal H})\GITquot G$ is precisely the geometry described by this algebraic variety. We can motivate this by the following observations:

The invariant polynomials take the same values for any two points in ${\cal H}$ on the same $G$-orbit, so we have a map between orbits and $n$-tuples $(f_1,\dots,f_N)$. Thus, the invariant polynomials give a map from a $G$-orbit in ${\cal H}$ to $\mathbb{C}^N$, or from $G$-orbits in $\bbP({\cal H})$ to $\mathbb{C P}(k_1,\dots,k_N)$.

If a point $x_c \in {\cal H}$ is in the closure of the $G$-orbit of another point $x$, the invariant polynomials must agree for $x$ and $x_c$ since they are continuous functions on ${\cal H}$. Thus, we have in particular that
\begin{itemize}
\item
For any unstable point $x$ in ${\cal H}$, all invariant polynomials (generated by $f_i$ of positive degree) must vanish, since they vanish for $0$, which is in the closure of $G \cdot x$.
\item
For any two points on equivalent semistable orbits, all invariant polynomials will agree, since they must take the same values as for points on the common polystable orbit in their closures.
\end{itemize}
It turns out that the converse of the last statement is also true: if all the invariant polynomials agree, the two points lie in the same equivalence class. Thus, the map from $\bbP({\cal H})\GITquot G$ to the algebraic space defined by the invariant polynomials and their relations is an isomorphism, so the quotient has the structure of a (closed) algebraic
subvariety of the weighted projective space. Geometries defined in this way are guaranteed to be K\"ahler, i.e. the symplectic structure defined by viewing them as a symplectic quotient is compatible with the complex structure.

As a special case, we note that the set of LME states will be empty if and only if there are no $G$-invariant polynomials:
\begin{itemize}
\item
For dimensions $(d_1, \dots, d_n)$, there exist locally maximally entangled states if and only if there exist invariant polynomials of $G = \SL(d_1, \mathbb{C}) \times \cdots \times \SL(d_n, \mathbb{C})$
\end{itemize}

\subsubsection*{Representation theory criterion for the existence of locally maximally entangled states}

In terms of representation theory, the existence of an invariant polynomial is equivalent to the condition that the symmetrized $k$-th tensor power of the $(d_1,d_2,\cdots, d_n)$ representation of $\SL(d_1, \mathbb{C}) \times \cdots \times \SL(d_n, \mathbb{C})$ contains the trivial representation for some $k$.

This representation theory question is equivalent, via Schur-Weyl duality, to a question about representation theory of the symmetric group (see, for example, \cite{Walter}). We recall that representations of the symmetric group $S_n$ may be labelled by Young diagrams with $n$ boxes. For our specific case, the representations we are interested in are representations $R(A,B)$ corresponding to rectangular Young diagrams with $A$ rows and $B$ columns. The condition is that for some $k$, the tensor product
\be
\label{reps}
R(d_1, k d_2 \cdots d_n) \times R(d_2, k d_1 d_3 \cdots d_n) \times \cdots \times R(d_n, k d_1 \cdots d_{n-1})
\ee
contains the trivial representation.

These representation theory criteria also follow from general results about the compatibility of spectra of density matrices (known as the {\it quantum marginal problem}) \cite{Kly04}. The existence of locally maximally entangled states is equivalent to asking whether it is possible to find a density matrix $\rho_{A_1 \cdots A_n}$ with spectrum $(1,0,0,\dots,0)$ for which $\rho_{A_k}$ has spectrum $(1/d_k,1/d_k, \dots, 1/d_k)$.

Using the representation theory of the symmetric group, it is possible to come up with an explicit calculational check for when the product of representations (\ref{reps}) contains the trivial representation; see Exercise 4.51 and Theorem 4.10 in Fulton and Harris~\cite{fh}. For example, the number of trivial representations in the tensor product for the case $n=3$ is given by
\be
\label{check}
N = \sum_p \frac{1}{z(p)} \chi_{(d_1, k d_2 d_3)}(C_p) \chi_{(d_2, k d_1 d_3)}(C_p) \chi_{(d_3, k d_1 d_2)}(C_p)
\ee
where the sum runs over partitions of $D = k d_1 d_2 d_3$, labeled by Young diagrams with $i_k$ columns of length $k$  such that $\sum_k k i_k = D$,
\be
z(p) = i_1!1^{i_1} i_2! 2^{i_2} \cdots i_D! D^{i_D} \; ,
\ee
and $\chi_R(C_p)$ is the character associated with representation $R$, evaluated for the conjugacy class $C_p$ associated with $p$. Frobenius gave an explicit formula for this (see~\cite[4.10]{fh}), so in principle, to decide if there is a locally maximally entangled state in the tensor product of Hilbert spaces with dimension $(d_1,d_2,d_3)$ we need only evaluate the expression \eqref{check} as a function of $k$ and check whether it is ever non-zero. Unfortunately, this turns out to be computationally hard for all but the smallest dimensions. Thus, we will instead make use of algebraic methods, to which which we turn in the next section.

\section{Understanding $\LME$ using geometric invariant theory}

In the previous section, we have reviewed how the space $\LME/K$ of LME states up to local unitary transformations is equivalent to the geometric invariant theory (GIT) quotient $\bbP({\cal H})\GITquot G$ of the full space of states by the group of local determinant-one invertible transformations $G = \SL(d_1,\mathbb{C}) \times \cdots \times \SL(d_n,\mathbb{C})$.

Starting from the latter description, it is possible to use the machinery of geometric invariant theory to characterize the space, providing explicit results that tell us for which $(d_1, d_2, \dots, d_n)$ LME states exist, and give the dimension of the space in all nonempty cases.

Our rigorous results characterizing the GIT quotient $\LME/K$ appear in a companion mathematics paper \cite{mathpaper}. In this section, our goal is to present an overview of the results and their derivation.

\subsubsection*{Dimensionality of $\bbP({\cal H})\GITquot G$}

To understand the dimension of the quotient $\bbP({\cal H})\GITquot G$, we note that the original space $\bbP({\cal H}) = \mathbb{CP}^{d_1 \cdots d_n -1}$ has complex dimension $d_1 d_2 \cdots d_n - 1$ while the group $G = \SL(d_1, \mathbb{C}) \times \cdots \times \SL(d_n, \mathbb{C})$ has complex dimension $\sum_i (d_i^2 -1)$. The latter is naively the dimension of a typical $G$-orbit, so the dimension of $\bbP({\cal H})\GITquot G$, the space of these orbits, is naively the difference
\be
\label{defDel1}
{\rm dim}(\bbP({\cal H})\GITquot G)_{naive} = {\rm dim}(\bbP({\cal H})) - {\rm dim}(G) = \left(\prod_{i=1}^n d_i - 1 \right) - \sum_{i=1}^n( d_i^2 -1) \equiv \Delta(d_1, \dots d_n)\; .
\ee
This naive dimension can fail to be correct, however, if a typical point in ${\cal H}$ is invariant under some subgroup of $G$ with positive dimension.

For any point $\psi \in {\cal H}$, we define $S_\psi \subset G$ to be the subgroup of $G$ that leaves $\psi$ invariant. This is called the {\it stabilizer} at position $\psi$. The dimension of $S_\psi$ is equal to the dimension of the subspace of infinitesimal transformations in the Lie algebra of $G$ that leaves $\psi$ invariant. As we review in \cite[\S3]{mathpaper},
there always exists some subgroup $S \subset G$ and a dense open subset
of $U \subset {\cal H}$ such that the stabilizer at any point in $U$
is conjugate to $S$. Moreover,~
\be
\label{realdim}
{\rm dim}(\bbP({\cal H})\GITquot G) = {\rm dim}(\bbP({\cal H})) - {\rm dim}(G) + {\rm dim}(S)    \; .
\ee
A key point here is that the group $G = \SL(d_1, \mathbb{C}) \times \cdots \times \SL(d_n, \mathbb{C})$
is semisimple.~\footnote{For a discussion of the relationship
between the dimension of the GIT quotient and the dimension of
the stabilizer for a generic state in the quantum information
literature, see, e.g.,~\cite{MOA}.}
If this dimension is $-1$, this means that a neighborhood of a generic point in the space ${\cal H}$ is contained within the $G$-orbit of this point. In this case, it follows that the quotient is empty, since a point $\lambda x$ for $\lambda < 1$ will then be in the $G$-orbit of $x$, and by iterating the group action that gives this point, we will end up arbitrarily close to 0. Thus, the orbit of a generic point is unstable.

For specific examples in which the dimensions are not too large, it is straightforward to directly calculate ${\rm dim}(S)$ and thus the dimension of the quotient.  For some fixed point $\psi$, an element
\be
g = g_1 \otimes \identity \otimes \cdots \otimes \identity + \dots + \identity \otimes \cdots \otimes \identity \otimes g_n
\ee
of the Lie algebra of $G$ (where $g_i$ is a $d_i \times d_i$ traceless matrix) will be a generator of the stabilizer group if and only if $g \cdot \psi = 0$. This gives ${\rm dim}{\cal H}$ linear equations for the ${\rm dim}G$ variables (the elements of $g_i$). We define ${\cal R}$ to be the number of these equations that are independent (i.e. the rank of the rectangular matrix that defines the linear equations) for a generic point $\psi$ - this gives the dimension of the generic orbit in ${\cal H}$. This may be computed by evaluating the rank in the case where the coefficients $\psi$ are left as undetermined variables. Then
\be
{\rm dim} S = {\rm dim}G - {\cal R}
\ee
and
\be
{\rm dim}(\bbP({\cal H})\GITquot G) = {\rm dim}(\bbP({\cal H})) - {\cal R} \; .
\ee
Implementing this method using Maple, we have calculated the dimension in many specific examples and found that the results agree in all cases with the predictions of Theorems \ref{thm.R} and \ref{thm.dims}.

\subsection{Proof of theorems governing existence and dimension of $\LME$}

We are now ready to outline a proof of the key results Theorems \ref{thm.R} and \ref{thm.dims} for the existence and dimension of $\LME/K$. We refer to the companion paper \cite{mathpaper} for the complete proof. Here, the aim is to prove as much as possible with minimal mathematical background, and give an overview of the remaining steps.

Both theorems follow by analyzing the following recursive result for the dimension of $\LME/K$:

\begin{theorem}\label{thm.States}
Consider a multipart quantum system described by Hilbert space ${\cal H} = {\cal H}_1  \otimes \cdots  \otimes {\cal H}_n$ with subsystems ${\cal H}_i$ of dimension $d_i$, with  $d_1 \leqslant d_2 \leqslant \cdots \leqslant d_n$. Let $\LME/K$ be the space of locally maximally entangled states up to identification by local unitary transformations $K = \SU(d_1) \times \dots \times \SU(d_n)$. Then
\begin{enumerate}

\item
If $d_n > d_1 \cdots d_{n-1}$, then $\LME/K$ is empty.

\item
If $d_n = d_1 \cdots d_{n-1}$, then $\LME/K$ is a single point.

\item
If $d_n \le \dfrac{1}{2} d_1 \cdots d_{n-1}$,
then $\LME/K$ is non-empty.
It is a single point in the case where $n = 3$ and $d_1 = d_2 = d_3 = 2$,
has dimension $d - 3$ in the case where $n = 3$ and $d_1 = 2$, $d_2 = d_3 = d$
and has dimension $d_1 \cdots d_n -
d_1^2 - \dots - d_n^2 + n - 1$ in all other cases.

\item
If $\dfrac{1}{2} d_1 \cdots d_{n-1} < d_n < d_1 \cdots d_{n-1}$,
then $\LME/K$ has the same dimension as the space $\LME/K$ for dimensions $\{d_1,\dots,d_{n-1}, d_1 \cdots d_{n-1} - d_n\}$.
\end{enumerate}
\end{theorem}
In the last case, the sum of dimensions after the transformation is strictly less than the original set of dimensions, so the recursion terminates after a finite number of steps.

\subsubsection*{Outline of proof of theorem \ref{thm.States}}

\begin{itemize}
\item
{\bf Case 1: $d_n > d_1 \cdots d_{n-1}$}

In this case, we violate the necessary conditions (\ref{necessary}), so there cannot be any LME states, and the GIT quotient must be empty.
\item
{\bf Case 2: $d_n = d_1 \cdots d_{n-1}$}

In this case, the Hilbert space is a tensor product of two subsystems of dimension $d_n = d_1 \cdots d_{n-1}$. Since the density matrices for these subsystems have the same spectrum, having the elementary $d_n$-dimensional subsystem maximally mixed implies that the composite subsystem composed of the first $n-1$ elementary subsystems is also maximally mixed. This implies that each density matrix for the first $n-1$ subsystems is maximally mixed, so our state is LME if and only if the two complementary $d_n$ dimensional subsystems are maximally mixed. Describing the state as a $d_n \times d_n$ matrix $\psi_{i I}$, the condition that both subsystems are maximally mixed is that $\psi^\dagger \psi = \psi \psi^\dagger = \dfrac{1}{ d_n} \identity$; this holds if and only if $\psi_{i I} = \dfrac{1}{\sqrt{d_n}} U_{i I}$ where $U$ is a unitary matrix. By an $\SU(d_n)$ transformation on the $d_n$-dimensional elementary subsystem and an overall phase rotation, we can bring the state to the form $\psi_{i I} = \dfrac{1}{\sqrt{d_n}} \delta_{i I}$. This is just the Bell state (\ref{Bell}), so for this case, we have a unique LME state up to local unitary transformations, as claimed.

\item

{\bf Case 3: $d_n \le \dfrac{1}{2} d_1 \cdots d_{n-1}$}

In this case, it is not hard to show that the naive dimension $\Delta(d_1, d_2, \dots d_n)$ is positive except in the special case where $n=3$ and $(d_1, d_2, d_3) = (2,d,d)$. In this latter case, we have already shown by our explicit construction in section 2 that the space of LME states up to local unitaries is a single point for $d=2$ and has dimension $d-3$ for $d > 3$. In the remaining cases, we show in \cite{mathpaper} making use of the work of
Elashvili \cite{elashvili}  that the stabilizer for a generic point has dimension 0. Thus, by the result (\ref{realdim}), the dimension is the naive dimension $\Delta(d_1, d_2, \dots d_n) > 0$ and the quotient is non-empty.

\item
{\bf Case 4: $\dfrac{1}{2} d_1 \cdots d_{n-1} < d_n <  d_1 \cdots d_{n-1}$}

For this case, we claim that the dimension of the quotient for $(d_1, \dots , d_n)$ is the same as the dimension of the quotient for $(d_1, \dots ,d_{n-1}, d_1 \cdots d_{n-1} - d_n)$. To show this, we recall that the geometry of the quotient may be specified by describing a set of generators for the $G$-invariant polynomials in the coordinates $\psi$ describing the state, together with their relations. Letting $D_0 = d_1 \cdots d_{n-1}$ and $G_0 = \SL(d_1, \mathbb{C}) \times \cdots \times \SL(d_n, \mathbb{C})$, we will now argue that the ring of polynomials on $\mathbb{C}^{D_0}  \otimes \mathbb{C}^{d_n}$ invariant under $G_0 \times \SL(d_n, \mathbb{C})$ is isomorphic to the ring of polynomials on $\mathbb{C}^{D_0}  \otimes \mathbb{C}^{{D_0} - d_n}$ invariant under $G_0 \times \SL({D_0} - d_n, \mathbb{C})$.

To see this, let $\psi_{I i_n}$ represent coordinates on $\mathbb{C}^{D_0}  \otimes \mathbb{C}^{d_n}$ where $1 \le i_n \le d_n$ and $I$ represents the $(n-1)$-tuple $(i_1, \dots, i_{n-1})$. Any holomorphic polynomial in $\psi_{I i_n}$ invariant under $\SL(d_n, \mathbb{C})$ must have all $i_n$ indices contracted with the $\SL(d_n, \mathbb{C})$ invariant antisymmetric tensor $\epsilon^{i_1 \cdots i_{d_n}}$. Thus, the polynomials invariant under $G_0 \times \SL(d_n, \mathbb{C})$ are $G_0$-invariant polynomials built from the $\SL(d_n, \mathbb{C})$ invariant coordinate
\be
P_{I_1 \cdots I_{d_n}} \equiv \epsilon^{i_1 \cdots i_{d_n}} \psi_{I_1 i_1} \cdots \psi_{I_{d_n} i_{d_n}} \; .
\ee
Similarly, if $\psi_{I i_n}$ give coordinates on $\mathbb{C}^{D_0}  \otimes \mathbb{C}^{{D_0} - d_n}$, the polynomials invariant under $G_0 \times \SL({D_0} - d_n, \mathbb{C})$ are $G_0$-invariant polynomials built from the $\SL({D_0} - d_n, \mathbb{C})$ invariant coordinate
\be
Q_{J_1 \cdots J_{{D_0} - d_n}} \equiv \epsilon^{i_1 \cdots i_{{D_0} - d_n}} \psi_{J_1 i_1} \cdots \psi_{J_{{D_0} - d_n}  i_{D_0 - d_n}} \; .
\ee
But using the $\SL({D_0},\mathbb{C})$-invariant $\epsilon$ tensor, we can alternatively define coordinates
\be
\hat{P}^{I_1 \cdots I_{d_n}} \equiv \epsilon^{I_1 \cdots I_{d_n} J_1 \cdots J_{{D_0}-d_n}} Q_{J_1 \cdots J_{{D_0} - d_n}} \; .
\ee
These define a space isomorphic to the space defined by $P_{I_1 \cdots I_{d_n}}$, and the $G_0$ action on the two spaces is equivalent, so the space defined by the $G_0$-invariant polynomials in $P$ and their relations has the same dimension as the space of $G_0$-invariant polynomials in $\hat{P}$. Thus, the dimension of the quotient $\bbP({\cal H})\GITquot G$ is the same for dimensions $(d_1, \dots , d_n)$ and $(d_1, \dots , d_1 \cdots d_{n-1} - d_n)$, as claimed.
\end{itemize}

\subsubsection*{Outline of proofs for Theorems \ref{thm.R} and \ref{thm.dims}}

The recursive algorithm of Theorem~\ref{thm.States} can be used to tell us which $(d_1, \dots, d_n)$ admit LME states and to
compute the dimension of the space $\LME/K$.

In this section, we explain how this leads to the more direct results in Theorems \ref{thm.R} and \ref{thm.dims}.

Since the recursive step in Theorem~\ref{thm.States} (case 4) always decreases the sum of dimensions, it will always terminate on one of the first three cases after a finite number of steps.

We can predict which will be the terminal case by making use of various quantities that are invariant under the transformation
\be
\label{recursion}
\{d_1,\dots,d_{n-1}, d_n\} \rightarrow \{d_1,\dots,d_{n-1}, d_1 \cdots d_{n-1} - d_n\}
\ee
in case 4 of Theorem~\ref{thm.States}.

Direct calculation shows that the naive dimension  $\Delta(d_1, \dots, d_n)$ defined in (\ref{defDel1}) is invariant under (\ref{recursion}). Additionally, the set of greatest common divisors of $k$-tuples of the dimensions is invariant for any $2 \le k \le n$.

A particularly useful combination of these invariants is the quantity
\be
R(\vec{d}) = \prod_i d_i + \sum_{k=1}^n (-1)^k \sum_{1 \le i_1 < \cdots < i_k \le n} (\gcd(d_{i_1},\dots,d_{i_n}))^2 \; .
\ee
We can show (see Proposition 5.3 in \cite{mathpaper} for a detailed proof) that $R$ is always negative, zero, or positive for dimensions satisfying the conditions of case (1), (2), or (3) of Theorem~\ref{thm.States} respectively. Since $R$ is invariant under the recursion step of case (4), it follows that
\begin{proposition}\label{prop:Rcases}
The algorithm in Theorem 1.2 terminates on case (1) if and only if $R<0$, on case (2) if and only if $R=0$ and on case (3) if and only if $R>0$.
\end{proposition}
Combining this result with Theorem~\ref{thm.States} immediately gives Theorem \ref{thm.R}.

To demonstrate Theorem \ref{thm.dims}, we consider also the behavior of the invariants $\Delta(d_1, \dots d_n)$ and
\be
\label{gmd2}
{\rm g}_{max}(\vec{d}) \equiv \max_{1 \le i < j \le n} \gcd(d_i,d_j)
\ee
for the terminal cases. We find (see Proposition 6.1 in \cite{mathpaper} for a detailed proof)
\begin{proposition} \label{prop.terminal}
For $d_1 \le \cdots \le d_n$, define $\Delta(d_1, \dots, d_n)$ as in (\ref{defDelta}) and let $(d'_1, \dots, d'_n)$ be the dimensions reached for the terminal case.
\begin{enumerate}
\item
If $\Delta(d_1, \dots d_n) < -2$, the recursion in Theorem \ref{thm.States} terminates on cases (1) or (2). Then $\LME/K$ is empty or a single point.
\item
If $\Delta(d_1, \dots d_n) = -2$, the recursion in Theorem \ref{thm.States} terminates on case (3) with $(d'_1, \dots d'_n) = (1,\dots,1,2,d',d')$ where $d' = g_{max}(\vec{d})$ as defined in (\ref{gmd}). Then $\LME/K$ is a point for $d'=2$ and has dimension $d'-3$ otherwise.
\item
If $\Delta(d_1, \dots d_n) > -2$, the recursion in Theorem \ref{thm.States} terminates on case (3) with $(d'_1, \dots d'_n) \ne (1,\dots,1,2,d',d')$. Then $\LME/K$ has dimension $\Delta(d_1, \dots d_n) > 0$.
\end{enumerate}
\end{proposition}
Combining this proposition with Proposition \ref{prop:Rcases} and Theorem \ref{thm.R} implies Theorem \ref{thm.dims}.

\subsection{Explicit results}

In this section, we analyze Theorems \ref{thm.R} and \ref{thm.dims} to provide more explicit results in some cases.

\subsubsection*{Behavior of $\LME$ as a function of $d_n$ for fixed $(d_1, \dots, d_{n-1})$}

It is straightforward to characterize the behavior of $\LME/K$ as we vary $d_n$ for some fixed $(d_1,d_2,\dots,d_{n-1})$. We assume that $n > 3$ or $n=3$ and $d_1 > 2$, since we understood the cases $n=2$ and $n=3, d_1=2$ in detail in section 2.

For these remaining cases, we note that as a function of $d_n$ for fixed $(d_1,d_2,\dots,d_{n-1})$,  $\Delta(d_1,\dots,d_n)$ is a downwards parabola that has a positive value for $d_n = d_{n-1}$, increases to a maximum at $d_n = d_1 \cdots d_{n-1}/2$ and then decreases. Let $d_n = d_*$ be the value at which $\Delta(d_n)$ reaches -2. Explicitly,
\[
d_* = \frac{P}{2} + \frac{1}{2}\sqrt{P^2 - 4 Q + 8} \qquad  P = \prod_{i=1}^{n-1} d_i,\;  Q = \sum_{i=1}^{n-1} (d_i^2 - 1) \; .
\]
Then according to Proposition~\ref{prop.terminal}, for $d_{n-1} \le d_n < d_*$ the space $\LME/K$ is non-empty and has dimension $\Delta(d_1, \dots d_n)$.  If $d_*$ is an integer, $\LME/K$ is nonempty for $d_n= d_*$; in this case, the recursion terminates at $(1,\dots,1,2,d',d')$ and $\LME/K$ has dimension $d'-3$ for $d' \ge 3$ and $0$ for $d' = 2$. Finally, for $d_* \le d_n \le d_1 \cdots d_{n-1}$, there is a set of sporadic cases for which $\LME/K$ is a single point.

For $n=3$ (including the $n_1=2$ case) we can characterize these sporadic cases (including $d_n = d_*$ cases) explicitly.

\begin{proposition}\label{prop.Fib}
Let $n=3$ and suppose the quotient is nonempty for $(d_1,d_2,d_3)$ with $\Delta(d_1,d_2,d_3) < 0$ or equivalently $d_3 \ge d_*(d_1,d_2)$. Then we have $(d_1,d_2,d_3) = (A, f_i, f_{i+1})$ where $(f_i,f_{i+1})$ are successive terms in a sequence defined by
\be
\label{Fib}
f_{i+1} = A f_i - f_{i-1}
\ee
with $(f_0,f_1) = (k,kA)$ for positive integer $k$ or $(f_0,f_1) \in S_A$. Here $S_A$ is a set of pairs $(b,c)$ defined by the requirement that
\be
\label{Sregion}
 b \le  \frac{A}{2} c \qquad  c \le \frac{A}{2} b  \qquad bc \ge A \qquad Abc - A^2 - b^2 - c^2 + 4 \le 0  \; .
\ee
and that the quotient corresponding to $(b,c,A)$ is nonempty.
\end{proposition}
\begin{proof}
By part (3) of Proposition~\ref{prop.terminal}, $\Delta < 0$ implies $\Delta \le -2$; this is also implied by $d_3 \ge d_*(d_1,d_2)$ from the definition of $d_*$. Thus, suppose that the quotient is non-empty for dimensions $(A,d_2,d_3)$ with $\Delta(A,d_2,d_3) \le -2$. Consider applying the operation
\be
\gamma: (A,B,C) \to (A, A B - C, B)
\ee
repeatedly {\it without reordering} to define a sequence $((A,B_0,C_0) \equiv (A,d_2,d_3),(A,B_1,C_1), \dots)$. Since we must have $C_0 \ge A_0 B_0/2$ for $\Delta < 0$, the initial step does not increase the sum of the elements. Further, all elements remain positive unless we end up on a triple for which $C_i = A B_i$. Thus, repeating the operation $\gamma$ must either bring us to a triple $(A,b,Ab)$ or to a triple $(A,b,c)$ for which the operation $\gamma$ does not decrease the sum of the elements.

In the latter case, we can show that $b$ and $c$ satisfy
\be
 b \le  \frac{A}{2} c \qquad  c \le \frac{A}{2} b \qquad bc \ge A \qquad Abc - A^2 - b^2 - c^2 + 4 \le 0  \; .
\ee
The first pair of inequalities for $c$ come from demanding that $\gamma$ applied to $(A,b,c)$ does not decrease the sum and that $(A,b,c)$ came by acting with $\gamma$ on a triple with a sum that was not smaller. The third inequality follows since we are assuming the quotient is not empty. The fourth inequality is the statement that $\Delta \le -2$.

If $(A,B,C)$ descends via $\gamma$ to $(A,b,c)$, the quotient corresponding to this triple must be non-empty, so the pair $(b,c)$ is in the set $S_A$ defined by the proposition. Thus, any triple $(A,d_2,d_3)$ with $\Delta \le -2$ descends via the recursion step to $(A,k,Ak)$ or $(A,b,c)$ with $(b,c) \in S_A$. To determine the form of such triples $(A,d_2,d_3)$ explicitly, note that the inverse of the operation $\gamma$ is the operation
\be
B_{i+1} = C_i \qquad C_{i+1} = A C_i - B_i
\ee
Combining these we find that $B_i$ and $C_i$ must be successive terms in the sequence
\be
f_{i+1} = A f_i - f_{i-1} \; ,
\ee
where the allowed starting values are $(f_0,f_1) = (k,Ak)$ or $(f_0, f_1) \in S_A$.
\end{proof}
\begin{remark}
The general terms in the sequence defined by (\ref{Fib}) can be given explicitly via a generating function as
\be
f_i = \frac{f_0 + (f_1 - A f_0) x}{1 - Ax + x^2}|_{x^n} \; .
\ee
For $A \ge 3$, the region defined by~\eqref{Sregion} covers a narrow band of the plane near the curve $bc=A$, symmetric about the line $b=c$ and contained in the region $b < A$, $c < A$. Thus, while the definition of $S_A$ still involves the recursion from Theorem~\ref{thm.States}, we need only check a finite number of points (on the order of $A$) in the region \eqref{Sregion} to determine the set, after which we can use Proposition~\ref{prop.Fib} to give an explicit formula for all $\Delta < 0$ triples $(A,d_2,d_2)$ with a nonempty quotient. As examples, we have
\bea
S_3 &=& {(3,2),(2,2),(2,3)} \cr
S_4 &=& {(4,2),(3,2),(2,3),(2,4)} \cr
S_5 &=& {(5,2),(4,2),(2,4),(2,5)}
\eea
We consider the special case $A=2$ presently.
\end{remark}
\begin{example}
For dimensions $(2,d_2,d_3)$ the quotient $\mathbb \bbP({\cal H}) \GITquot  G$ is non-empty if and only if $(d_1,d_2,d_3) = (2,b,b)$ for $b \ge 2$ or $(2,kb,(k+1)b)$ for positive integers $k,b$ with $kb > 1$. The quotient is a single point except for $(d_1,d_2,d_3) = (2,b,b)$ with $b>3$, in which case it has dimension $b-3$.
\end{example}
\begin{proof}
For this case, the range $d_{n-1} \le d_n < d_*$ is empty, so the only non-empty quotients are those covered by Proposition~\ref{prop.Fib}. For $A=2$, the definition (\ref{Fib}) of the sequence in Proposition~\ref{prop.Fib} may be written as $f_{i+1} - f_{i} = f_i - f_{i-1}$ so the sequence is arithmetic. The conditions (\ref{Sregion}) imply $b=c$, so we find $S_2 = \{(b,b)|b \ge 2\}$. The proposition then gives that the non-empty triples are $(2,f_i,f_{i+1})$ with $(f_0,f_1) = (b,2b)$ or $(f_0,f_1) = (b,b)$. Since the sequence is arithmetic, we have explicitly that $(2,b,b)$ with $b \ge 2$ and $(2,kb,(k+1)b)$ with $bk > 2$. Proposition~\ref{prop.terminal}, implies that the quotient is a point except in the case $(2,b,b)$, where it has dimension $b-3$.
\end{proof}

This reproduces the results that we obtained by our explicit construction in section 2.

\section*{Acknowledgements}

We would like to thank Jason Bell, Patrick Hayden, Alex May, Robert Raussendorf, David Stephen, and Michael Walter for helpful comments and discussions. This work is supported in part by the Natural Sciences and Engineering Research Council of Canada and by the Simons Foundation.

\appendix

\section{Explicit construction of all LME states for $(d_1,d_2,d_3) = (2,B,C)$}

In this appendix, we provide details of the explicit construction of all LME states with $(d_1, d_2, d_3) = (2,B,C)$ with $2 \le B \le C$. In this case, the necessary conditions (\ref{necessary}) require that
\be
B \le C \le 2 B \; .
\ee

Using the Schmidt decomposition, and performing a $U(2)$ rotation on the first factor, any state $|\Psi \rangle$ with our desired properties can be written as
\be
|\Psi \rangle = \frac{1}{\sqrt{2}} | 1 \rangle  \otimes \psi^1_{bc} |b\rangle  \otimes |c \rangle + \frac{1}{\sqrt{2}} | 2 \rangle  \otimes \psi^2_{bc} |b\rangle  \otimes |c \rangle
\ee
where $\psi^1_{bc} |b\rangle  \otimes |c \rangle$ and $\psi^2_{bc} |b\rangle  \otimes |c \rangle$ define orthonormal states of ${\cal H}_B  \otimes {\cal H}_C$ and summation over $b$ and $c$ is implied.

Making use of the Schmidt decomposition on $\psi^1_{bc} |b\rangle  \otimes |c \rangle$ and performing $U(B)$ and $U(C)$ rotations on the second and third factors, we can write
\be
\psi^1_{bc} = {\large \left[D_{\{\sqrt{p_i}\}} \, \, \,  0_{B \times (C-B)} \right] \qquad
\psi^2_{bc} = \left[{\cal I}_{B \times B} \, \, \, {\cal J}_{B \times (C-B)} \right]}
\ee
where $D_{\{\sqrt{p_i}\}}$ is the diagonal matrix with elements $\sqrt{p_i}$. Orthogonality of $\psi^1$ and $\psi^2$ requires that
\be
\label{cond1}
\tr({\cal I}D_{\{ \sqrt{p_i} \}}) = 0 \; .
\ee
The condition that $\rho_B$ is maximally mixed gives
\be
\label{cond2}
\frac{1}{2} D_{\{p_i\}} + \frac{1}{2} {\cal I} {\cal I}^\dagger + \frac{1}{2} {\cal J} {\cal J}^\dagger = \frac{1}{B} \identity_{B \times B}
\ee
while the condition that $\rho_C$ is maximally mixed gives
\bea
\label{cond3}
\frac{1}{2} D_{\{p_i\}} + \frac{1}{2} {\cal I}^\dagger {\cal I}  &=& \frac{1}{C} \identity_{B \times B} \\
\label{cond4}
{\cal I}^\dagger {\cal J} &=& 0 \\
\frac{1}{2} {\cal J}^\dagger {\cal J} &=& \frac{1}{C} \identity_{(C-B) \times (C-B)} \; .
\label{cond5}
\eea

\subsubsection*{Case $C=B$}

We begin with the special case $C=B$. Here, the conditions collapse to
\be
\label{cond6}
{\cal I} {\cal I}^\dagger = {\cal I}^\dagger {\cal I} = D_{\{\frac{2}{B} - p_i\}}
\ee
together with the normalization condition (\ref{cond1}). Defining Hermitian matrices
\be
H_+ = \frac{1}{2} ({\cal I} + {\cal I}^\dagger) \qquad \qquad H_- = - \frac{i}{2} ({\cal I} - {\cal I}^\dagger)
\ee
the first equality in (\ref{cond6}) gives $[H_+,H_-] = 0$, so the matrices are simultaneously diagonalizable. We can write
\be
H_1 = U D_1 U^\dagger \qquad H_2 = U D_2 U^\dagger
\ee
so that
\be
{\cal I} = U D_{\{z_i\}} U^\dagger
\ee
where $D_{\{z_i\}}$ is some general complex diagonal matrix. The latter equality in (\ref{cond6}) gives that
\be
U D_{\{|z_i|^2\}} U^\dagger = D_{\{\frac{2}{B} - p_i\}} \; .
\ee
Without loss of generality, we can assume that the eigenvalues $p_i$ are ordered from largest to smallest and the eigenvalues in $D_{\{|z_i|^2\}}$ are ordered from smallest to largest. Then we must have
\be
z_i = e^{i \phi_i} \sqrt{\frac{2}{B} - p_i}
\ee
and $U$ must commute with $D_{\{p_i\}}$. It is therefore block-diagonal, with blocks corresponding to blocks of equal eigenvalues in $D_{\{p_i\}}$. However, recalling that the local unitary transformations on the two subsystems of size $B$ act as
\be
D_{\{p_i\}} \to W^\dagger D_{\{p_i\}} V \qquad \qquad {\cal I} \to W^\dagger {\cal I} V
\ee
we see that whenever such blocks exist, we can take local unitary transformations with $W = V = U$, to eliminate $U$, leaving
\be
{\cal I} = D_{\{e^{i \phi_i} \sqrt{\frac{2}{B} - p_i}  \}}
\ee
Whenever $p_i=0$, we have residual local unitary transformations that can be used to set $\phi_i=0$ also. We can also set $\phi_i = 0$ when $p_i = \dfrac{2}{B}$.

Finally, the condition (\ref{cond1}) gives that
\be
\label{polygon}
\sum_i e^{i \phi_i} \sqrt{p_i(\frac{2}{B} - p_i)}  = 0
\ee
so we must choose the phases so that the complex numbers in the sum add to zero. This implies that the quantities $\sqrt{p_i(\frac{2}{B} - p_i)}$ must satisfy polygon inequalities.

With these constraints, we can write the most general locally maximally entangled state up to local unitary transformations as
\be
\label{2BBstate}
|\Psi \rangle = \frac{1}{\sqrt{B}} \cos(\theta_i /2) |1 \rangle  \otimes |i \rangle  \otimes |i \rangle + \frac{1}{\sqrt{B}} e^{i \phi_i} \sin(\theta_i /2) |2 \rangle  \otimes |i \rangle  \otimes |i \rangle  \; .
\ee
where we have made the change of variables
\be
p_i = \frac{2}{B} \cos^2(\theta_i/2)
\ee
with $0 \le \theta_i \le \pi$.
With this parametrization, the constraint $\sum p_i = 1$ together with the constraint (\ref{polygon}) give that
\bea
\sum_i \cos(\theta_i) &=& 0 \cr
\sum_i \sin(\theta_i) \cos(\phi_i) &=& 0 \cr
\sum_i \sin(\theta_i) \sin(\phi_i) &=& 0
\eea
Thus it is natural to think of $(\theta_i, \phi_i)$ as spherical coordinates defining $B$ unit vectors $\vec{x}_i = (\cos(\theta_i),\sin(\theta_i) \cos(\phi_i),\sin(\theta_i) \sin(\phi_i))$ in  $\mathbb{R}^3$ which must add to zero. It is not hard to check (e.g. by studying the action of the algebra) that the any two sets of such vectors related by a rotation in $SO(3)$ define equivalent states; they are related by performing an $\SU(2)$ rotation on the first factor followed by a transformations in the diagonal subgroup of $\SU(B) \times \SU(B)$ that put the state back in the form (\ref{2BBstate}). In summary, for the case $(d_1,d_2,d_3) = (2,B,B)$ the space of locally maximally entangled states up to local unitary transformations is equivalent to the space of sets of $B$ unit vectors in $\mathbb{R}^3$ adding to zero modulo simultaneous $SO(3)$ rotations of the vectors. This space has complex dimension $B-3$ for $B \ge 3$ and is a point (the orbit of the GHZ state (\ref{GHZ})) for $B=2$.

\subsubsection*{Case $C>B$}

Next, we consider the remaining case $C > B$. Starting from condition (\ref{cond5}), we must have that ${\cal J}$ takes the form
\be
\sqrt{\frac{2}{C}} \left( v_1 \cdots v_{C-B} \right)
\ee
where $v_i$ are a set of orthonormal vectors in ${\cal H}_B$. Choosing $2B-C$ vectors $\hat{v}_i$ to complete the basis of ${\cal H}_B$, the condition (\ref{cond4}) implies that the columns of ${\cal I}$ are each some linear combinations of the vectors $\hat{v}_i$. Defining the unitary matrix
\be
U = \left( v_1 \cdots v_{C-B} \;  \hat{v}_1 \cdots \hat{v}_{2B - C} \right)
\ee
and defining
\be
R_+ = \left( \ba{c} \identity_{(C-B) \times (C - B)} \cr 0_{(2B - C) \times (C - B)} \ea \right) \qquad R_- = \left( \ba{c} 0_{(C-B) \times (2B - C)} \cr \identity_{(2B - C) \times (2B - C)} \ea \right)
\ee
we can write
\be
{\cal J} = \frac{\sqrt{2}{C}} U R_+ \qquad \qquad {\cal I} = U R_- M
\ee
where $M$ is some $(2B - C) \times B$ matrix. Now, condition (\ref{cond3}) gives that
\be
\label{MM}
M^\dagger M = D_{\{\frac{2}{C}-p_i\}}.
\ee
Since the rank on the left side can be at most $2B -C$, we must have at least $C -B$ of the $p_i$s equal to $2/C$. We assume that these are the first $C-B$ and denote the remaining ones by $p_i'$. These must all be less than or equal to $2/C$ since $M^\dagger M$ cannot have a negative eigenvalue. Since only the lower right $(2B -C) \times (2B - C)$ block of $D$ is nonvanishing, the constraint (\ref{MM}) is solved by
\be
M = \left( 0_{(2B-C) \times (C-B)} \qquad u D_{\sqrt{2/C-p'_i}} \right)
\ee
where $u$ is a $(2B-C) \times (2B - C)$ unitary matrix.

Now, from condition (\ref{cond2}), we get that
\be
\label{Ueq}
U \left[ \ba{cccc} \frac{2}{C} & & & \cr & \ddots & & \cr & & \frac{2}{C} & \cr & & & u D_{\{2/C - p'_i\}} u^\dagger \ea \right] U^\dagger = D_{\{2/B - p_i\}}
\ee
where the upper left block of the matrix comes from the terms involving ${\cal J}$ and the lower right block comes from the terms involving ${\cal I}$. Now, we can assume that the $p_i$ are ordered largest to smallest, with the first $C-B$ equal to $2/C$ and the rest labeled as $p'_1,p'_2,\dots$ as we argued above. The eigenvalues of the matrix on the right, ordered smallest to largest are
\be
\left( \frac{2}{ B} -\frac{2 }{ C}, \dots, \frac{2 }{ B} - \frac{2 }{ C}, \frac{2 }{ B} - p'_1, \dots, \frac{2 }{ B} - p'_{2B-C} \right)
\ee
The eigenvalues of the matrix on the left, ordered from smallest to largest, are
\be
\left( \frac{2 }{ C} - p'_1, \dots , \frac{2 }{ C} - p'_{2B-C}, \frac{2 }{ C}, \dots, \frac{2 }{ C} \right)
\ee
These two ordered sets must be equal, so we can solve for all the eigenvalues in terms of $B$ and $C$ by setting the first $2B-C$ entries equal in the two sets. However, the last $C-B$ entries must also agree, and this gives a constraint: we must have that $C - B$ divides $B$. In this case, we can write $B = N K$ and $C = (N+1) K$, and the eigenvalues are
\be
\frac{2 }{ (N+1)K} \left(1, \dots, 1, \frac{N-1 }{ N}, \dots, \frac{N-1 }{ N},\dots, \frac{1 }{ N} \dots \frac{1}{ N}\right)
\ee
with each eigenvalue appearing $C-B = K$ times. From (\ref{Ueq}), we then require that
\be
\label{Ueq2}
U  = \left[ \ba{cccc} U_1 & & & \cr & U_2 & & \cr & & \ddots & \cr & & & U_N \ea \right] \left[\ba{c} R_-^\dagger \cr  R_+^\dagger \ea \right]\left[ \ba{cc} \identity & 0  \cr 0 & u^\dagger \ea \right]
\ee
Here, the rightmost matrix removes the $u$ in (\ref{Ueq}), the second matrix is a permutation that places the eigenvalues in the same order as on the right hand side, and the leftmost matrix is a block diagonal matrix of $K \times K$ unitaries $U_i$ that leave the matrix invariant. Combining everything, we find the general solution for the case $(B,C) = (NK,(N+1)K)$ is
\be
\left[{\cal I} \, \, \, {\cal J} \right] = \sqrt{\frac{2}{(N+1) K}} \left[ \ba{cccc} U_1 & & & \cr & U_2 & & \cr & & \ddots & \cr & & & U_N \ea \right] \left[ {\large 0}_{NK \times K} \ba{cccc}  \sqrt{\frac{1}{N}} \identity_K & & & \cr & \sqrt{\frac{2}{N}} \identity_K & & \cr & & \ddots & \cr & & & \identity_K \ea \right]
\ee
However, we can set all the $U_i$ to $\identity$ by residual $U(B) \times U(C)$ transformations $V \times W$ where
\bea
V &=& \left[ \ba{cccc} W_1 & & & \cr & W_2 & & \cr & & \ddots & \cr & & & W_n \ea \right] \cr
W &=& \left[ \ba{cccc} W_1^* & & & \cr & W_2^* & & \cr & & \ddots & \cr & & & W_{n+1}^* \ea \right] \cr
W_1 &=& \identity \qquad W_{i > 1} = U_1 U_2 \cdots U_i
\eea
Thus, for $(B,C) = (nK,(n+1)K)$, we have a unique locally maximally entangled state up to local unitary transformations, given by
\be
|\Psi_{(2,NK,(N+1)K)} \rangle = \frac{1}{\sqrt{(N+1)K}} \sum_{i=1}^K \sum_{b=1}^N \left\{ \sqrt{N + 1 - \frac{b}{N}} |0 \rangle  \otimes  |b \; i \rangle  \otimes |b \; i \rangle + \sqrt{\frac{b}{N}}  |1 \rangle  \otimes|b \; i \rangle  \otimes |b+1 \; i \rangle \right\}
\ee
This is simply a tensor product
\be
|\Psi_{(2,N,(N+1))} \rangle  \otimes |\Psi_{(K,K)} \rangle
\ee
where $|\Psi_{(K,K)} \rangle$ is the Bell state
\be
|\Psi_{(K,K)} \rangle = \frac{1}{\sqrt{K}}  \sum_{i=1}^K  |i \rangle  \otimes  |i \rangle \; .
\ee

\subsection{Sudoku States for $n=3$}

As an aside, we provide an additional explicit construction for LME states with $n=3$, related to the puzzle game Sudoku. Consider any $B \times C$ grid in which we place the numbers $1, \dots, A$, each appearing $k$ times and each appearing at most once in each row and at most once in each column, and such that each row and column are occupied the same number of times ($kA/B$ and $kA/C$ times respectively). Solutions to standard Sudoku puzzles give an example for $A=B=C=k=9$. We can express this grid of numbers as a matrix
\be
M = 1 \cdot M_1 + 2 \cdot M_2 + \dots  + A \cdot M_A
\ee
where each $M_i$ is a $B \times C$ matrix with each element $0$ or $1$. Then, constructing a quantum state
\be
|\Psi \rangle = \sum_{i=1}^A \frac{1}{\sqrt{kA}} (M_i)_{bc} |i \rangle  \otimes |b \rangle  \otimes |c \rangle \; ,
\ee
we can see that it will be locally maximally entangled.

\section{Representation theory construction of LME states for $(d_1,d_2,d_3) = (2,p,p)$ with prime $p$}

In this appendix, we demonstrate that for a tripartite quantum system with subsystems of dimensions $(d_1,d_2,d_3) = (2,p,p)$ for prime $p$, we can construct LME states using the representation theory construction of section 3. We use a finite group $H$ that is a particular semidirect product of $\mathbb{Z}_2$ with $UT(3,p)$, a finite group of order $p^3$ that can be represented by matrices
\begin{equation}
\label{UT}
\begin{pmatrix}
1 & a & b \\
0 & 1 & c \\
0 & 0 & 1
\end{pmatrix}
,~a,b,c \in \mathbb{Z}_p \; .
\end{equation}

\subsection{Irreducible representations of $UT(3,p)$}

The group $UT(3,p)$ has $(p-1)$ irreducible representations of dimension $p$ and $p^2$ irreducible representations of dimension 1, that we now describe.

\subsubsection*{One dimensional irreducible representations of $UT(3,p)$}
The $p^2$ one-dimensional irreducible representations of $UT(3,p)$ are given by
\be
\label{RUT}
\begin{pmatrix}
1 & a & b \\
0 & 1 & c \\
0 & 0 & 1
\end{pmatrix}
\to e^{\frac{2 \pi i}{p} (x a + y c)} \qquad \qquad x,y \in \mathbb{Z}_p
\ee
Here, we have the trivial representation for $x=y=0$ and $p^2-1$ nontrivial one-dimensional irreducible representations.

\subsubsection*{$p$ dimensional irreducible representations of $UT(3,p)$}

The $p-1$ $p$-dimensional irreducible representations can be labeled by a parameter $y \in \big(\mathbb{Z}_p\big)^{\times}$. We can describe these explicitly by
\begin{equation}
\begin{gathered}
R_y: UT(3,p) \rightarrow \GL(p, \mathbb{C}),\\
\begin{pmatrix}
1 & a & b \\
0 & 1 & c \\
0 & 0 & 1
\end{pmatrix}
\mapsto M \qquad ~M_{jk} = \delta_{j \; [k - a]} e^{\frac{2 \pi i}{p} y (c(k-a)+b)}
\end{gathered}
\end{equation}
where the matrix index in square brackets is interpreted modulo $p$.

\subsection{The group $H = UT(3,p) \rtimes_{\phi}~\mathbb{Z}_2$}

The group used in our construction is a semidirect product $UT(3,p) \rtimes_{\phi}~\mathbb{Z}_2$. To define this, let $\{1,s\}$ be the elements of $\mathbb{Z}_2$ with $s^2 = 1$, and label the elements of $UT(3,p) \rtimes_{\phi}~\mathbb{Z}_2$ by $(h,t)$ where $h \in UT(3,p)$ and $t \in \mathbb{Z}_2$. To specify the multiplication rule for $H$, we define
the group homomorphism $\phi: \mathbb{Z}_2 \rightarrow Aut(UT(3,p))$ given by $\phi(1) = \identity,~\phi(s) = \phi_{s}$ with
\begin{equation}
\phi_{s} :
\begin{pmatrix}
1 & a & b \\
0 & 1 & c \\
0 & 0 & 1
\end{pmatrix}
\rightarrow
\begin{pmatrix}
1 & -a & b \\
0 & 1 & -c \\
0 & 0 & 1
\end{pmatrix}.
\end{equation}
It is straightforward to verify that $\phi_s$ is indeed an automorphism of $UT(3,p)$ that preserves elements in the center of $UT(3,p)$. In terms of $\phi$, the group multiplication rule for $H$ is then specified as
\begin{equation}
(h,~t) \cdot (h',~t') = (h \phi_t (h'),~tt').
\end{equation}

\subsubsection*{Conjugacy class structure for $UT(3,p) \rtimes \mathbb{Z}_2$}

The group $UT(3,p) \rtimes \mathbb{Z}_2$ has some conjugacy classes `inherited' from the smaller group, $UT(3,p)$ and then $p$ large conjugacy classes outside of this normal subgroup. All together there are $\frac{p^2-1}{2}$ classes of size $2p$ and $p$ classes of size $1$ making up the $UT(3,p)$ subgroup, and $p$ classes of size $p^2$ making up its complement.

\vspace{2mm}
In order to understand conjugation, we note that for $g = (h,t) \in UT(3,p) \rtimes \mathbb{Z}_2$, we have that
\begin{equation}
g^{-1}  = (\phi_t (h^{-1}),~t) \; ,
\end{equation}
using the fact that $\phi_t$ is an involution.

The general conjugation formula of an element $g = (h,~t) \in H$ by the element $g' = (h',~t') \in H$ is then
\begin{equation}
\begin{gathered}
(h',~t')(h,~t)(h',~t')^{-1}  = (h'\phi_{t'}(h)\phi_t(h'^{-1}),~t).
\end{gathered}
\label{eq:generalconj}
\end{equation}
For $g$ in the $UT(3,p)$ subgroup, this gives
\begin{equation}
g'gg'^{-1} = (h'\phi_{t'}(h)h'^{-1},~e).
\end{equation}
Thus, it is clear that conjugacy classes of elements in the {\it normal subgroup} $UT(3,p)$ are almost unchanged with respect to those of $UT(3,p)$ as an independent group, except that now elements become conjugate to their images under the automorphism $\phi_s$. From the structure of $\phi_s$ it is clear that only central elements ($a=c=0$) are fixed under this automorphism. We thus find that the elements of the subgroup $UT(3,p)$ split into $\frac{p^2-1}{2}$ conjugacy classes (labelled by $(a,~c) \neq (0,~0) \in (\mathbb{Z}_p \times \mathbb{Z}_p)/\{\pm 1\}$) of size $2p$ comprising the sets:
\begin{equation}
\bigg\{ \bigg(
\begin{pmatrix}
1 & \pm a & b \\
0 & 1 & \pm c \\
0 & 0 & 1
\end{pmatrix},~e \bigg)~|~b \in \mathbb{Z}_p \bigg\},
\end{equation}
and $p$ classes (labelled by $b \in \mathbb{Z}_p$) of size one comprising the sets:
\begin{equation}
\bigg\{ \bigg(
\begin{pmatrix}
1 & 0 & b \\
0 & 1 & 0 \\
0 & 0 & 1
\end{pmatrix},~e \bigg) \bigg\}.
\end{equation}
In the case that the element which we are conjugating lies in the complement of the $UT(3,p)$ subgroup, equation (\ref{eq:generalconj}) tells us that we remain in this set. Taking $h \in UT(3,p)$ of the form (\ref{UT}) we obtain the following result for the $UT(3,p)$ piece:
\begin{equation}
\begin{gathered}
h'\phi_{t'}(h)\phi_s(h'^{-1}) =
\begin{pmatrix}
1 & x & y \\
0 & 1 & z \\
0 & 0 & 1
\end{pmatrix}
\begin{pmatrix}
1 & \pm a & b \\
0 & 1 & \pm c \\
0 & 0 & 1
\end{pmatrix}
\begin{pmatrix}
1 & x & -y + xz \\
0 & 1 & z \\
0 & 0 & 1
\end{pmatrix} \\
=
\begin{pmatrix}
1 & \pm a + 2x & b + 2xz \pm az \pm cx \\
0 & 1 & \pm c + 2z \\
0 & 0 & 1
\end{pmatrix}
\end{gathered},
\label{eq:matrixconj}
\end{equation}
with the upper signs in the case $t' = e$ and the lower signs if $t' = s$.

From this, we can see that by choosing $x$ and $z$ appropriately, we may conjugate to an element in the centre of $UT(3,p)$. Moreover, elements in $H$ of the form $(h_c,~s)$ with $h_c \in Z(UT(3,p))$ are not conjugate to one another (as we see by setting $a=c=0$ and noticing that the corresponding element in equation \ref{eq:matrixconj} is central in $UT(3,p)$ if and only if $x=z=0$) and thus may be used to label the $p$ conjugacy classes of size $p^2$ that make up the complement to $UT(3,p)$ in $H$. We can write these explicitly as
\begin{equation}
\bigg\{ \bigg(
\begin{pmatrix}
1 & 2x & b + 2xz \\
0 & 1 & 2z \\
0 & 0 & 1
\end{pmatrix},~s \bigg)~|~(x,~z) \in \mathbb{Z}_p \times \mathbb{Z}_p \bigg\}.
\end{equation}

\subsubsection*{Representatives of conjugacy classes from group generators}

Letting $a_{ij}$ denote the element of $H$ given by $(h_{ij},~e)$, with $h_{ij}$ being the 3 by 3 unitriangular matrix with all zeroes on the off-diagonals except for a $1$ in the $i^{th}$ row and $j^{th}$ column, and $a_{\rtimes}$ denote $(\identity,~s) \in H$, we obtain more compact notation for the group generators. Namely, they are $a_{\rtimes},~a_{12},~a_{23},~a_{13}$. Then:
\begin{itemize}
\item
 $a_{12}^a a_{23}^c$ is a representative in the size $2p$ conjugacy class (in $UT(3,p)$) labelled by $(a,~c) \neq (0,~0) \in (\mathbb{Z}_p \times \mathbb{Z}_p)/\{\pm 1\}$.
\item
$a_{13}^b$ gives the size $1$ conjugacy class (in $UT(3,p)$) labelled by $b \in \mathbb{Z}_p$
representative .
\item
$a_{13}^ba_{\rtimes}$ is a representative in the size $p^2$ conjugacy class (in the complement of $UT(3,p)$ in $H$) labelled by $b \in \mathbb{Z}_p$.
\end{itemize}

\subsection{Irreducible representations of $UT(3,p) \rtimes \mathbb{Z}_2$}

Let us now describe the irreducible representations of $UT(3,p) \rtimes \mathbb{Z}_2$.

\subsubsection*{One-dimensional irreducible representations of $UT(3,p) \rtimes \mathbb{Z}_2$}

The two one-dimensional irreps come from taking $R((\identity,~s)) = \pm 1$, with all generators of $UT(3,p)$ mapping to $1$. In the trivial representation, the generator of $\mathbb{Z}_2$ maps to $1$ and in the non-trivial 1D irrep it maps to $-1$.

\subsubsection*{Two-dimensional irreducible representations of $UT(3,p) \rtimes \mathbb{Z}_2$}

The $\frac{p^2 - 1}{2}$ two-dimensional irreps come from taking a 1D irrep $R_{UT}$ of $UT(3,p)$ (of which there are $p^2 - 1$ non-trivial ones) and extending it to $H$ by defining:
\begin{equation}
R((\identity,s)) =
\begin{pmatrix}
0 & 1 \\
1 & 0
\end{pmatrix},~R((h,1)) =
\begin{pmatrix}
R_{UT} (h) & 0 \\
0 & R^{-1}_{UT} (h))
\end{pmatrix},\\
\end{equation}
Using (\ref{RUT}), we can write explicitly that
\begin{equation}
R \bigg( \bigg(
\begin{pmatrix}
1 & a & b \\
0 & 1 & c \\
0 & 0 & 1
\end{pmatrix},~e \bigg) \bigg) =
\begin{pmatrix}
e^{\frac{2 \pi i}{p} (x a + y c)} & 0 \\
0 & e^{-\frac{2 \pi i}{p} (x a + y c)}
\end{pmatrix}
\end{equation}
We denote these two-dimensional irreps by $2^{x,y}$, where $(x,y) \ne (0,0)$ and $2^{x,y} \equiv 2^{-x,-y}$.

\subsubsection*{Characters of two-dimensional irreducible representations}

For the size $2p$ class labelled by $(a,~c)$ the character for the representation $2^{x,y}$ is $\chi_{a,c} = e^{\frac{2 \pi i}{p} (x a + y c)} + e^{-\frac{2 \pi i}{p} (x a + y c)}$.

For the size $1$ classes labelled by $b$ the character is $\chi_b = 2$, as all these elements map to identity.

For the size $p^2$ classes labelled by $\tilde{b}$ the character is $\chi_{\tilde{b}} = 0$, since the final step in reaching this conjugacy class is multiplying a diagonal matrix by $\sigma_x$, which yields a traceless (purely off-diagonal) matrix.

\subsubsection*{$p$-dimensional irreducible representations of $UT(3,p) \rtimes \mathbb{Z}_2$}

The $2(p-1)$ $p$-dimensional irreps of $H$ come from taking a p-dimensional irrep of $UT(3,p)$ ($p-1$ of them), $R_{UT}$, and using this for the generators of $UT(3,p)$ in $H$. For the element $(\identity,~s) \in H$, we have two choices, given by
\begin{equation}
(R((\identity,s)))_{mn} = \pm \delta_{n[p-m]},
\end{equation}
where $m,~n \in \{0,1,...,p-1 \}$ label the rows and columns, and we interpret these labels modulo $p$.

Recall that $R_{UT}$ is labelled by some $y \in Z_p^{\times}$ and gives matrices with elements:
\begin{equation}
\delta_{j \; [k - a]} e^{\frac{2 \pi i}{p} y (c(k-a)+b)}.
\end{equation}

We label these irreps as $p^{\pm,y}$ with y the parameter above and $\pm$ the sign in front of the $(\identity,~s)$ matrix.

\subsubsection*{Characters of $p$-dimensional irreducible representations}

For the size $2p$ class labelled by $(a,~c)$ the character is $\chi_{a,c} = 0$, since if $h_{12} = a \neq 0$ then the matrix is purely off-diagonal. If $h_{12} = a = 0$, then the trace of the representing matrix is proportional to the sum of all $p$ of the $p^{th}$ roots of unity, which is zero.

For the size $1$ classes labelled by $b$ the character is $\chi_b = pe^{\frac{2 \pi i}{p} b y}$, since the representing matrix is simply $e^{\frac{2 \pi i}{p} b y}$ times the $p$ by $p$ identity matrix.

For the size $p^2$ classes labelled by $\tilde{b}$ the character is $\chi_{\tilde{b}} = \pm e^{\frac{2 \pi i}{p} \tilde{b} y}$, since the final step in reaching this conjugacy class leaves only the leftmost column and top row element on the diagonal (with an overall sign given by the sign of the matrix representing the generator of $\mathbb{Z}_2$), which is $\big(e^{\frac{2 \pi i}{p} y}\big)^{\tilde{b}} = e^{\frac{2 \pi i}{p} \tilde{b} y}$.

\subsubsection*{Tensor Product of 2-dimensional and p-dimensional irreducible representations of $UT(3,p) \rtimes \mathbb{Z}_2$}

Multiplying the characters of $2^{x,y}$ and $p^{\pm,y'}$, we find that the result only has nonzero characters (equal to $2pe^{\frac{2 \pi i}{p} b y'}$) on the center of the group $H$. Looking at the characters of the $p$-dimensional irreps, one immediately sees that this can be understood as the sum of the characters of two $p$-dimensional representations with the opposite choice of sign for the matrix representing the generator of $\mathbb{Z}_2$. Hence we obtain:
\begin{equation}
2^{x,y}  \otimes p^{\pm,y'} = p^{+,y'} \oplus p^{-,y'}.
\label{singletensorpdt}
\end{equation}
It follows immediately that the tensor product of representations $2^{x,y}$, $p^{\pm,y'}$, and the dual of $p^{+,y'}$ or $p^{-,y'}$ will contain the trivial representation, so there exists a representation theory construction of LME states for the case $(d_1,d_2,d_3) = (2,p,p)$ as we wished to show.

\section{Multipart systems for which the generic state has trivial stabilizer}

In this section, we would like to understand for which multipart systems a generic state has trivial stabilizer.

First, we note that according to Theorem 2 of \cite{ampopov}, in the case where $d_n \le \frac{1}{2} d_1 \cdots d_{n-1}$ the stabilizer will be trivial unless $(d_1, \dots, d_n)$ is $(2,d,d)$, $(2,2,2,2)$ or $(3,3,3)$. These cases have $\Delta = -2$, $\Delta = 3$, and $\Delta = 2$ respectively.

Suppose first that for dimensions $(d_1,d_2, \dots, d_n)$, we have $\Delta > 3$. Then according to Proposition 5.3, the recursion in Theorem 5.1 will terminate on case (3) and we have $d_n' \le \dfrac{1}{2} d_1' \cdots d_{n-1}'$. Since $\Delta > 3$, we can't be in any of the special cases above, so the stabilizer for the terminal dimensions $(d_1',d_2', \dots, d_n')$ is trivial. According to Corollary 1 to Lemma 2 of \cite{elashvili}, the property that the stabilizer of a generic point is trivial is preserved under the recursion step of Theorem 5.1. Thus, we conclude that for $\Delta > 3$, the stabilizer of a generic point is trivial.

Conversely, suppose that the stabilizer of a generic point is trivial for some dimensions $(d_1,d_2, \dots, d_n)$. By (62), this requires that the actual dimension of the quotient $\bbP({\cal H})\GITquot G$ equals the expected dimension $\Delta$. Since the actual dimension cannot be less than or equal to -2, we are in case 3 of Proposition 5.3, so $\Delta > 0$ and the recursion in Theorem 5.1 terminates on case (3) with $d_n' \le \dfrac{1}{2} d_1' \cdots d_{n-1}'$  and $(d_1', \dots, d_n') \ne (1,\dots,1,2,d',d')$. For $\Delta > 0$ at least three of the $d_i'$s must be greater than 1. Since $\Delta(\vec{d})$ is unchanged if we remove any 1s, the conditions of the following Lemma are satisfied for $\{ d_i' \}$ with 1s removed. Since the cases $(2,2,2,2)$ and $(3,3,3)$ do not have trivial stabilizer, the Lemma implies that $\Delta > 3$.

\begin{lemma}
Suppose $2 \le d_1 \le \cdots \le d_n \le \dfrac{1}{2} d_1 \cdots d_{n-1}$ for $n \ge 3$ and $\Delta(\vec{d}) > 0$. Then $(d_1, \dots, d_n) = (3,3,3)$ with $\Delta = 2$, $(d_1, \dots, d_n) = (2,2,2,2)$ with $\Delta = 3$, or $\Delta > 3$.
\end{lemma}
\begin{proof}
As we explained in section 5.2, $\Delta$ considered as a function of $d_n$ is a downward parabola that increases monotonically from $d_n = d_{n-1}$ to its maximum at $d_n = \dfrac{1}{2} d_1 \cdots d_{n-1}$. Thus, for $d_n \le \dfrac{1}{2} d_1 \cdots d_{n-1}$ we have that
\beas
\Delta(d_1, \dots, d_n) &\ge& \Delta(d_1, \dots, d_{n-1}, d_{n-1}) \cr
 &=& d_1 \cdots d_{n-1} d_{n-1} - d_1^2 - \cdots - d_{n-1}^2 - d_{n-1}^2 + n - 1 \cr
 &\ge& (d_1 \cdots d_{n-2} - n)d_{n-1}^2  + n - 1 \cr
 &\ge& (2^{n-2} - n)d_{n-1}^2  + n - 1
\eeas
For $n > 4$, the expression in parentheses in the last line is positive, so we have $\Delta > 3$.

For $n=4$, the expression in parentheses in the last line vanishes, so we have that $\Delta \ge 3$. For equality, the last inequality requires $d_1 = d_2 = 2$, the penultimate inequality requires $d_1 = d_2 = d_3$, and the first inequality requires that $d_4 = d_3$. Thus we have $\Delta > 3$ or $(d_1, \dots, d_4) = (2,2,2,2)$ with $\Delta = 3$ in this case.

For $n=3$, $\Delta > 0$ requires that $d_1 \ge 3$ since for $(d_1,d_2,d_3) = (2,d_2,d_3)$, the maximum $\Delta$, achieved for $d_3 = d_2$, is -2. From the penultimate line above, the expression in parenthesis will then be positive (and consequently $\Delta$ will be at least 18) unless $d_1 = 3$. For $d_1 = 3$, the second line above gives $\Delta \ge (d_2^2 - 9) + 2$. This will be at least 9 unless $d_2 = 3$. Finally, for $(d_1,d_2,d_3) = (3,3,d_3)$, we find that $\Delta = 2$ for the case $(3,3,3)$ and $\Delta > 3$ otherwise.
\end{proof}

\bibliographystyle{unstr}

\end{document}